\documentclass[a4paper,10pt]{article}
\usepackage{amsmath,amssymb,graphicx}

\renewcommand{\title}[1]{\vspace{\fill}
\eject\addtolength{\baselineskip}{4pt}
{\bfseries\LARGE #1}\\[3mm]\addtolength{\baselineskip}{-4pt}}
\renewcommand{\author}[3]{\parbox[t]{75mm}
{\begin{center}{\scshape #1}\\[3mm] #2\\
 {\ttfamily #3} \end{center}}}

\newtheorem{thm}{\bfseries Theorem}
\newtheorem{lem}[thm]{\bfseries Lemma}
\newtheorem{remark}[thm]{\bfseries Remark}

\newtheorem{prop}[thm]{\bfseries Proposition}

\newtheorem{defn}[thm]{\bfseries Definition}

\newtheorem{conj}[thm]{\bfseries Conjecture}

\newenvironment{proof}{\medskip
\noindent{\scshape Proof:}}{\quad $\Box$\medskip}
\usepackage{braket}
\usepackage{cleveref}
\usepackage[hang,small,bf]{caption}
\usepackage[subrefformat=parens]{subcaption}
\usepackage[ruled]{algorithm2e}
\usepackage{todonotes}

\begin{document}

\begin{center}

\title{Quantum-Relaxation Based Optimization Algorithms: Theoretical Extensions} 
\author{\underline{Kosei Teramoto}
}{
Department of Computer Science, \\
The University of Tokyo
}{
teramoto@is.s.u-tokyo.ac.jp
}
\author{Rudy Raymond
}{
IBM Quantum, IBM Japan \\
Department of Computer Science, \\
The University of Tokyo \\
Quantum Computing Center, Keio University
}{
rudyhar@jp.ibm.com
}
\author{
Eyuri Wakakuwa
}{ 
Department of Computer Science, \\
The University of Tokyo
}{
eyuriwakakuwa@is.s.u-tokyo.ac.jp
}
\author{
Hiroshi Imai
}{ 
Department of Computer Science, \\
The University of Tokyo
}{
imai@is.s.u-tokyo.ac.jp
}
\end{center}

\begin{quote}
{\bfseries Abstract:}
Quantum Random Access Optimizer (QRAO) is a quantum-relaxation based optimization algorithm proposed by Fuller et al. that utilizes Quantum Random Access Code (QRAC) to encode multiple variables of binary optimization in a single qubit.
The approximation ratio bound of QRAO for the maximum cut problem is $0.555$ if the bit-to-qubit compression ratio is $3$x, while it is $0.625$ if the compression ratio is $2$x, thus demonstrating a trade-off between space efficiency and approximability.
In this research, we extend the quantum-relaxation by using another QRAC which encodes three classical bits into two qubits (the bit-to-qubit compression ratio is $1.5$x) and obtain its approximation ratio for the maximum cut problem as $0.722$.
Also, we design a novel quantum relaxation that always guarantees a $2$x bit-to-qubit compression ratio which is unlike the original quantum relaxation of Fuller~et~al. We analyze the condition when it has a non-trivial approximation ratio bound $\left(>\frac{1}{2}\right)$.
We hope that our results lead to the analysis of the quantum approximability and practical efficiency of the quantum-relaxation based approaches.
\end{quote}

\begin{quote}
\end{quote}
\vspace{5mm}

\section{Introduction}
\subsection{Backgrounds}
Solving optimization problems is one of the most important tasks for which quantum computation is expected to be useful.
Various quantum algorithms have been devised for NP-hard optimization problems such as QAOA (Quantum Approximate Optimization Algorithms)~\cite{farhi2014quantum} proposed by Farhi, Goldstone, and Gutmann, and VQE (Variational Quantum Eigensolver)~\cite{peruzzo2014variational} proposed by Peruzzo et al. Although QAOA and VQE are classical-quantum hybrid algorithms designed for near-term devices capable of running only shallow circuits, there are some critical issues. The first issue is scalability. Because QAOA and VQE encode one classical bit into one qubit and the number of qubits of near-term quantum devices is at most several hundred qubits, the problem instance sizes are highly limited. The second issue is that we do not know if \textit{quantumness} (i.e. quantum entanglement) of constant-depth QAOA and VQE can give rise to a better result than the classical optimization algorithms, as indicated in~\cite{nannicini2019performance}. In other words, for combinatorial optimization, QAOA and VQE may not be attractive to be run on a quantum computer in the first place.

Recently, a new classical-quantum hybrid optimization algorithm, QRAO (Quantum Random Access Optimization)~\cite{fuller2021approximate} was proposed by Fuller et al. to address the above issues.
Specifically, the QRAO encodes multiple classical bits (less than or equal to three) into one qubit using the $(3,1)$-QRAC (Quantum Random Access Code)~\cite{ambainis2002dense,hayashi20064}.
Here, $(m,n)$-QRAC means the quantum random access codes which encode $m$ classical bits into $n$ qubits.
Due to this constant-factor improvement in scalability, Fuller et al. were able to perform experiments with QRAO on superconducting quantum devices to solve the largest instances of a maximum cut problem (up to 40 nodes using only 15 qubits).
Also, since QRAO searches for quantum states that correspond to solutions to the relaxation problem rather than classical solutions, the quantum state that is eventually discovered is an entangled state that cannot be directly interpreted as a classical solution.
Because of this, the methods like QRAO are called \textit{quantum-relaxation} and have been extended for more general quadratic programs~\cite{https://doi.org/10.48550/arxiv.2301.01778}. 
To obtain the classical solution, \textit{quantum state rounding} of the relaxed solution must be performed. 
Therefore, compared to standard VQE methods, QRAO may benefit from quantum entanglement if the entangled states result in better relaxed values. In other words, QRAO is inherently different from standard quantum-classical hybrid algorithms like QAOA and may benefit from quantum mechanical properties.
There exists an experimental result that there are some instances for which entanglement helps QRAO find optimal solutions~\cite{teramoto2023role}.

The quantum state rounding algorithm (\textit{magic state rounding}) used in QRAO is inspired by Goemans and Williamson's approximation algorithm for the maximum cut problem with an approximation ratio of $0.879$~\cite{goemans1995improved}. 
It randomly chooses the pair of two-bit-inverted relationships and decodes the encoded bits into one of the two candidates by performing the corresponding quantum measurement.
By quantum information theoretic analysis, it is proved that the approximation ratio of quantum-relaxation using $(3,1)$-QRAC is $0.555$ and that of quantum-relaxation using $(2,1)$-QRAC is $0.625$~\cite{fuller2021approximate}. While the optimality of standard QAOA or VQE is often assumed when the obtained quantum state is the ground state, the approximation ratios of QRAO are obtained regardless of the reachability of the ground state. Namely, the ratios are guaranteed as long as the relaxed value of the obtained quantum state exceeds that of the classical optimal value. This is crucial as finding the exact ground state can be extremely hard~\cite{kempe2006complexity}.  

The approximation ratios of $(3,1)$- and $(2,1)$-QRAC imply that the higher the space compression ratio the lower the approximation ratio is. There is a trade-off between space efficiency and approximability.
The approximation ratio bound of QRAO is much lower than Goemans and Williamson's $0.879$~\cite{goemans1995improved} which is proved to be optimal under the UGC (Unique Game Conjecture)~\cite{khot2007optimal}.
This is because the success probability of decoding each bit of the QRACs used in QRAO is not high.
The success probability of decoding each encoded bit is $\frac{1}{2}+\frac{1}{2\sqrt{2}}\approx0.85$ for $(2,1)$-QRAC and $\frac{1}{2}+\frac{1}{2\sqrt{3}}\approx0.79$ for $(3,1)$-QRAC~\cite{ambainis2002dense,hayashi20064}.

\subsection{Our Results}
In this research, we extend the quantum-relaxation in two ways: $(i)$ we introduce the use of $(3,2)$-QRAC to obtain a better approximation ratio with a slightly lower bit-to-qubit compression ratio, and  $(ii)$ we design a novel quantum-relaxation that always guarantees 2x bit-to-qubit compression ratio which is unlike the original quantum relaxation of Fuller et al. For $(i)$, we will show the formulation of the $(3,2)$-QRAC which encodes three classical bits into two qubits obtained by numerical calculation~\cite{imamichi2018constructions}.
The success probability of decoding each encoded bit is $\frac{1}{2}+\frac{1}{\sqrt{6}}\approx0.908$, and it is optimal among all $(3,2)$-QRACs based on the bound by Manv{\v{c}}inska and Storgaard~\cite{manvcinska2022geometry}.
Also, we extended the quantum relaxation by using this $(3,2)$-QRAC.
The instance of the problem is encoded into the problem Hamiltonian, and the maximum eigenstate of the Hamiltonian is explored.
By performing the quantum state rounding algorithm, we obtain the classical binary solution to the problem.
Furthermore, we proved the approximation ratio bound of the above quantum-relaxation based optimization algorithm for the MaxCut problem as $\frac{13}{18}\approx 0.722$.
The only assumption of the proof of the approximation ratio is the same as the one using $(3,1)$- or $(2,1)$-QRACs, that is, the energy of the found candidate quantum state for the maximum eigenstate of the problem Hamiltonian exceeds the optimum value of the original problem instance.
Although the space compression ratio of our quantum relaxation is $\frac{3}{2}=1.5$ and is lower than the one using $(3,1)$- or $(2,1)$-QRACs, the approximation ratio bound is better.
Our result is consistent with the trade-off between the space compression ratio and the approximability of the maximum cut problem.
Though the obtained approximation ratio bound $0.722$ is lower than that of Goemans and Williamson, the practical feasibility of quantum-relaxation based approaches is enhanced.

To always guarantee the bit-to-qubit compression ratio of QRAO using $(3,1)$-QRAC is essential as in the original QRAO the ratio becomes lower as the density of the graph instance increases.
This is because there is a constraint that the endpoints of each edge must be associated with different qubits.
For example, if the graph instance is the complete graph, then the number of qubits needed to run QRAO is the same as the number of vertices.
In such cases, the quantum-relaxation based optimizer has no space advantage against standard QAOA and VQE algorithms.
In this research, for $(ii)$, we propose new types of encoding which encode up to two classical bits into a single-qubit by using the $(3,1)$-QRAC.
The third encoded bit's position in $(3,1)$-QRAC corresponds to the parity of the two bits.
This modification allows us to remove the constraint that the endpoints of each edge have to be assigned to different qubits.
The space compression ratio of the algorithm is always 2x which is independent of the density of the graph instances.
Unfortunately, non-trivial approximation ratio bound $\left(>\frac{1}{2}\right)$ does not exist generally.
We calculate the approximation ratio of this new algorithm by using two parameters $\epsilon$ and $\lambda$ as $\max\left\{\frac{81-14\sqrt{3}+14\sqrt{3}\lambda+8\epsilon}{81+162\epsilon},\frac{27-14\lambda+12\epsilon}{27+54\epsilon}\right\}$.
The parameter $\epsilon$ is defined by the equation $\mathrm{OPT}=\left(\frac{1}{2}+\epsilon\right)|E|$ where $\mathrm{OPT}$ is the optimal cut value, and therefore $\epsilon$ quantifies the so-called MaxCutGain~\cite{CharikarWirth2004}.
The parameter $\lambda$ is the ratio of the edges whose endpoints are assigned to different qubits.
By using the approximation ratio bound, we analyze the condition of the graph instance that our algorithm gives a non-obvious approximation ratio bound for the maximum cut problem.
We hope that our results lead to the analysis of the quantum approximability and practical efficiency of the quantum-relaxation based approaches.

\section{Preliminaries}
\subsection{Basics of Quantum Computing}
A classical bit is either $0$ or $1$.
On the other hand, in quantum mechanics, a quantum bit (\textit{qubit}) is represented by a two-dimensional complex unit vector in a Hilbert space.
There are two basis vectors corresponding to the classical $0$ or $1$ as below.
\begin{equation}
    \ket{0}:=\begin{pmatrix}1\\0\end{pmatrix},\ket{1}:=\begin{pmatrix}0\\1\end{pmatrix}
\end{equation}
$\ket{0}$ is read as "ket $0$".
The counterpart expression is "bra", and it represents a Hermitian conjugate of "ket".
The state of a qubit then is expressed as a linear combination of the two basis vectors as
\begin{equation}
    \ket{\psi}:=\alpha\ket{0}+\beta\ket{1}=\begin{pmatrix}\alpha\\\beta\end{pmatrix},
    \label{eq:state}
\end{equation}
where the coefficients $\alpha$ and $\beta$ are complex numbers and and satisfy the condition $|\alpha|^2+|\beta|^2=1$.
The Kronecker product denoted by $\otimes$ is used to describe multiple qubit states. For example, the two qubits states made of two single qubit states $\ket{\psi}=\alpha\ket{0}+\beta\ket{1}$ and $\ket{\varphi}=\gamma\ket{0}+\delta\ket{1}$ are described as a unit vector in $\mathbb{C}^2\otimes\mathbb{C}^2=\mathbb{C}^4$ like the following.
\begin{equation}
    \ket{\psi}\otimes\ket{\varphi}=\alpha\gamma\ket{0}\otimes\ket{0}+\alpha\delta\ket{0}\otimes\ket{1}+\beta\gamma\ket{1}\otimes\ket{0}+\beta\delta\ket{1}\otimes\ket{1}=\begin{pmatrix}\alpha\gamma\\\alpha\delta\\\beta\gamma\\\beta\delta\end{pmatrix}
\end{equation}

Generally, $n$ qubits states are expressed as a normalized vector in the Hilbert space $(\mathbb{C}^2)^{\otimes n}=\mathbb{C}^{2^n}$.
For simplicity, we sometimes just write $\ket{b_1}\otimes\ket{b_2}\otimes\cdots\otimes\ket{b_n}$ as $\ket{b_1b_2\cdots b_n}$ and $\bra{b_1}\otimes\bra{b_2}\otimes\cdots\otimes\bra{b_n}$ as $\bra{b_1b_2\cdots b_n}$ where each $b_i\in \{0,1\}$.

Quantum operations to $n$ qubits state in closed quantum systems are described as $2^n\times 2^n$ Unitary matrices whose elements are complex numbers.
In a quantum circuit model, we describe quantum operations by using quantum gates.
The Pauli $X$, $Y$, and $Z$ gates (or operators) are single qubit gates and are used to represent the spin of single qubits.
\begin{equation}
    X=\begin{pmatrix}0&1\\1&0\end{pmatrix},Y=\begin{pmatrix}0&-i\\i&0\end{pmatrix},Z=\begin{pmatrix}1&0\\0&-1\end{pmatrix}
\end{equation}
The Pauli $X$ gate behaves like the NOT gate in classical circuits and it maps $\ket{0}$ to $\ket{1}$ and $\ket{1}$ to $\ket{0}$.
The Pauli $Y$ gate maps $\ket{0}$ to $i\ket{1}$ and $\ket{1}$ to $-i\ket{0}$. 
The Pauli $Z$ gate maps $\ket{0}$ to $\ket{0}$ and $\ket{1}$ to $-\ket{1}$.
The eigenvalues of the Pauli $X$, $Y$, $Z$ operators are $1$ and $-1$, and the corresponding eigenvectors (or eigenstates) are $\{\ket{+},\ket{-}\}$, $\{\ket{+i},\ket{-i}\}$, and $\{\ket{0},\ket{1}\}$ where
\begin{align}
    \ket{+}=\frac{1}{\sqrt{2}}\left(\ket{0}+\ket{1}\right)=\frac{1}{\sqrt{2}}\begin{pmatrix}1\\1\end{pmatrix},\ket{-}=\frac{1}{\sqrt{2}}\left(\ket{0}-\ket{1}\right)=\frac{1}{\sqrt{2}}\begin{pmatrix}1\\-1\end{pmatrix},\\
    \ket{+i}=\frac{1}{\sqrt{2}}\left(\ket{0}+i\ket{1}\right)=\frac{1}{\sqrt{2}}\begin{pmatrix}1\\i\end{pmatrix},\ket{-i}=\frac{1}{\sqrt{2}}\left(\ket{0}-i\ket{1}\right)=\frac{1}{\sqrt{2}}\begin{pmatrix}1\\-i\end{pmatrix}.
\end{align}

Generally, quantum measurements are defined as positive operator-valued measures (POVMs).
An operator $U$ is positive semidefinite if for all non-zero vector $\ket{v}$, $\bra{v}U\ket{v}\geq 0$.
A POVM consists of a set of positive semidefinite operators $\{M_a\}$ indexed by the measurement outcomes $a\in S$ satisfying the condition $\sum_a M_a=I$.
If we measure the quantum state $\ket{\psi}$ with the POVM $\{M_a\}$, then the probability that the measurement results is $a$ is given by
\begin{equation}
    \mathrm{Pr}(a):=\bra{\psi}M_a\ket{\psi},
\end{equation}
and the post-measurement state becomes
\begin{equation}
    \frac{\sqrt{M_a}\ket{\psi}}{\sqrt{\mathrm{Pr}(a)}}.
\end{equation}
It is known that a positive semidefinite operator has a unique square root of it $\sqrt{M_a}$.
A POVM $\{M_a\}$ is \textit{projective} if each operator $M_a$ satisfies the condition $M_a^2=M_a$.
We call such quantum measurements as \textit{projective measurements}.
If all of the operators $M_a$ of a projective measurement has matrix rank one, then it is called \textit{rank-1 projective measurement}.

We have formulated quantum mechanics using a vector in Hilbert space.
On the other hand, there is an alternative formulation, \textit{density operators}.
The density operator formalism is equivalent to the state vectors, but it's sometimes more convenient to describe the quantum system or observe the characteristics of the quantum state.
Consider the situation that the quantum state is one of the states $\ket{\psi}$ indexed by $i$ with probability $p_i$ respectively.
We call the set of the tuple of the probability and the state $\{p_i,\ket{\psi}\}$ as \textit{ensemble of pure states}.
The density operator for this system is defined by the equation:
\begin{equation}
    \rho:=\sum_ip_i\ket{\psi_i}\bra{\psi_i}.
\end{equation}
In the case of the state vectors, if we perform some unitary $U$ to the state $\ket{\psi}$, then the state is converted to $U\ket{\psi}$.
In the density operator formalism, the density operator $\rho$ is converted to $U\rho U^{\dag}$.
If we measure the state with the POVM $\{M_a\}$, then the probability that the measurement outcome is $a$ is
\begin{equation}
    \mathrm{Pr}(a):=\mathrm{Tr}[M_a\rho],
\end{equation}
and the post-measurement density operator is
\begin{equation}
    \frac{\sqrt{M_a}\rho\sqrt{M_a}^{\dag}}{\mathrm{Pr}(a)}.
\end{equation}
By using the density operator formulation, we can see a single qubit state differently.
It is known that a single qubit quantum state (that may be a mixed state) can be formulated like the following equation.
\begin{equation}
    \rho=\frac{1}{2}\left(I+r_xX+r_yY+r_zZ\right)
    \label{eq:qubit_density}
\end{equation}
where $r_x$, $r_y$, and $r_z$ are the real numbers satisfying the condition:
\begin{equation}
    r_x^2+r_y^2+r_z^2\leq 1.
\end{equation}

\subsection{Quantum Random Access Codes}
\label{sec:qrac}
The $n$ qubits are represented by a vector in $\mathbb{C}^{2^n}$ and seem to have much more information than the classical $n$ bits.
However, it is known that $n$ qubits are needed to transfer $n$-bit classical information without error by Holevo bound~\cite{holevo1979capacity}.
On the other hand, if we admit some errors, we can encode multiple classical bits into a single qubit by using $(n,1,p)$-QRA codes~\cite{ambainis2002dense}.
\begin{defn}[$(n,1,p)$-QRA codes~\cite{ambainis2002dense}]
    An $(n,1,p)$-QRA coding is a function that maps $n$-bit strings $x\in\{0,1\}^n$ to $1$-qubit states $\rho_x$ satisfying the following conditions that for every $i\in\{1,2,...,n\}$, there exists a POVM
    \begin{equation*}
        E^i=\{E^i_0,E^i_1\}
    \end{equation*}
    such that
    \begin{equation*}
        \mathrm{Tr}(E^i_{x_i}\rho_x)\geq p
    \end{equation*}
    for all $x\in\{0,1\}^n$, where $x_i$ is the $i$-the bit of $x$.
\end{defn}
The POVM $E^i$ corresponds to the decoding process.
By measuring the encoded state $\rho_x$ with the POVM $E^i$, we can decode the $i$-th encoded bits $x_i$ with probability $p$.
We noted that $(n,1,p)$-QRA codes is meaningless if $p\leq\frac{1}{2}$ because $p=\frac{1}{2}$ is equivalent to randomly choosing binary bits.
$(n,m,p)$-QRA coding for $m\geq 2$ can also be defined in the same way.
There exists $(2,1,0.85)$- and $(3,1,0.79)$-QRA codings~\cite{ambainis2002dense} which are used in QRAO~\cite{fuller2021approximate}.
\begin{prop}[$(2,1,0.85)$-QRA codes~\cite{ambainis2002dense}]
    Consider the map
    \begin{equation}
        (x_1,x_2)\mapsto\rho_{x_1,x_2}:=\frac{1}{2}\left(I+\frac{1}{\sqrt{2}}((-1)^{x_1}X+(-1)^{x_2}Z)\right).
        \label{eq:21qrac}
    \end{equation}
    For every pair of $(x_1,x_2)$, $\rho_{x_1,x_2}$ is a pure state and can be written in the form $\rho_{x_1,x_2}=\ket{\psi(x_1,x_2)}\bra{\psi(x_1,x_2)}$ where
    \begin{align*}
        \ket{\psi(0,0)}&=\cos{\frac{\pi}{8}}\ket{0}+\sin{\frac{\pi}{8}}\ket{1},
        \ket{\psi(0,1)}=\cos{\frac{3\pi}{8}}\ket{0}+\sin{\frac{3\pi}{8}}\ket{1}\\
        \ket{\psi(1,0)}&=\cos{\frac{5\pi}{8}}\ket{0}+\sin{\frac{5\pi}{8}}\ket{1},
        \ket{\psi(1,1)}=\cos{\frac{7\pi}{8}}\ket{0}+\sin{\frac{7\pi}{8}}\ket{1}
    \end{align*}
    Then, this map is a $(2,1,0.85)$-QRA coding with the POVMs (projective measurements, in fact):
    \begin{equation}
        E^1=\{\ket{+}\bra{+},\ket{-}\bra{-}\},E^2=\{\ket{0}\bra{0},\ket{1}\bra{1}\}.
        \label{eq:21povm}
    \end{equation}
    \label{prop:21qrac}
\end{prop}
The measurements in \Cref{eq:21povm} are the measurements in $X$ and computational basis.
The $X$ basis measurement is performed to decode the first classical bit while the computational basis measurement is performed to decode the second classical bit.
The $(2,1,0.85)$-QRA coding is visualized as vertices of the square on the $x$-$z$ plane in the Bloch sphere as shown in \Cref{fig:21bloch}.
\begin{prop}[$(3,1,0.79)$-QRA codes~\cite{ambainis2002dense,hayashi20064}]
    Consider the map
    \begin{equation}
        (x_1,x_2,x_3)\mapsto\rho_{x_1,x_2,x_3}:=\frac{1}{2}\left(I+\frac{1}{\sqrt{3}}((-1)^{x_1}X+(-1)^{x_2}Y+(-1)^{x_3}Z)\right).
        \label{eq:31qrac}
    \end{equation}
    For every pair of $(x_1,x_2,x_3)$, $\rho_{x_1,x_2,x_3}$ is a pure state and can be written in the form $\rho_{x_1,x_2,x_3}=\ket{\psi(x_1,x_2,x_3)}\bra{\psi(x_1,x_2,x_3)}$ where
    \begin{align*}
        \ket{\psi(0,0,0)}&=\cos{\tilde{\theta}}\ket{0}+e^{\frac{\pi i}{4}}\sin{\tilde{\theta}}\ket{1},\\
        \ket{\psi(0,0,1)}&=\sin{\tilde{\theta}}\ket{0}+e^{\frac{\pi i}{4}}\cos{\tilde{\theta}}\ket{1},\\
        \ket{\psi(0,1,0)}&=\cos{\tilde{\theta}}\ket{0}+e^{\frac{-\pi i}{4}}\sin{\tilde{\theta}}\ket{1},\\
        \ket{\psi(0,1,1)}&=\sin{\tilde{\theta}}\ket{0}+e^{\frac{-\pi i}{4}}\cos{\tilde{\theta}}\ket{1},\\
        \ket{\psi(1,0,0)}&=\cos{\tilde{\theta}}\ket{0}+e^{\frac{3\pi i}{4}}\sin{\tilde{\theta}}\ket{1},\\
        \ket{\psi(1,0,1)}&=\sin{\tilde{\theta}}\ket{0}+e^{\frac{3\pi i}{4}}\cos{\tilde{\theta}}\ket{1},\\
        \ket{\psi(1,1,0)}&=\cos{\tilde{\theta}}\ket{0}+e^{\frac{-3\pi i}{4}}\sin{\tilde{\theta}}\ket{1},\\
        \ket{\psi(1,1,1)}&=\sin{\tilde{\theta}}\ket{0}+e^{\frac{-3\pi i}{4}}\cos{\tilde{\theta}}\ket{1},
    \end{align*}
    where $\tilde{\theta}$ satisfies the condition $(\cos{\tilde{\theta}})^2=\frac{1}{2}+\frac{1}{2\sqrt{3}}>0.79$.
    Then, this map is a $(3,1,0.79)$-QRA codings with the POVMs (projective measurements, in fact):
    \begin{equation}
        E^1=\{\ket{+}\bra{+},\ket{-}\bra{-}\},E^2=\{\ket{+i}\bra{+i},\ket{-i}\bra{-i}\},E^3=\{\ket{0}\bra{0},\ket{1}\bra{1}\}.
        \label{eq:31povm}
    \end{equation}
    \label{prop:31qrac}
\end{prop}
The measurements in \Cref{eq:31povm} are the measurements in $X$, $Y$, and computational basis.
Each measurement is performed to decode the corresponding classical bit.
The $(2,1,0.85)$-QRA coding is visualized as vertices of the cube inscribed in the Bloch sphere as shown in \Cref{fig:31bloch}.
\begin{remark}
    To see \Cref{eq:21qrac,eq:31qrac}, we can formulate $(1,1,1)$-QRA codes like the following equation:
    \begin{equation}
        x_1\mapsto\rho_{x_1}:=\frac{1}{2}(I+(-1)^{x_1}Z).
        \label{eq:11qrac}
    \end{equation}
    Then, each encoded state is a pure state like $\rho_0=\ket{0}\bra{0}$ and $\rho_1=\ket{1}\bra{1}$.
    The corresponding POVM is just a computational basis measurement $\{\ket{0}\bra{0},\ket{1}\bra{1}\}$.
\end{remark}
The $(1,1,1)$-QRA coding is visualized as the bipolar points of the Bloch sphere as shown in \Cref{fig:11bloch}.
\begin{figure}[tb]
    \begin{tabular}{ccc}
        \begin{minipage}[b]{0.3\hsize}
            \centering
            \includegraphics[height=3.8cm]{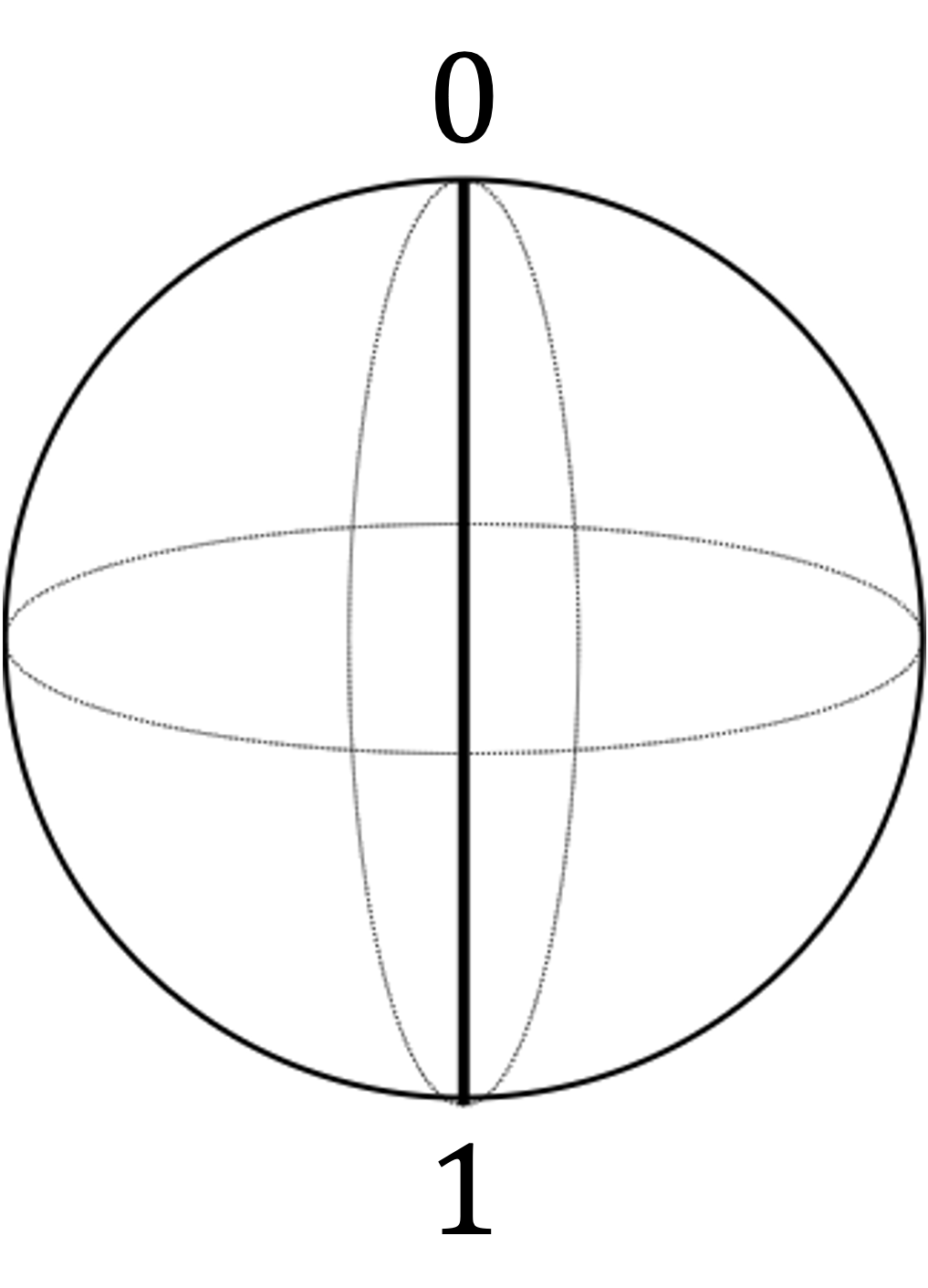}
            \subcaption{$(1,1,1)$-QRA coding}       
            \label{fig:11bloch}
        \end{minipage} &
        \begin{minipage}[b]{0.3\hsize}
            \centering
            \includegraphics[height=3.8cm]{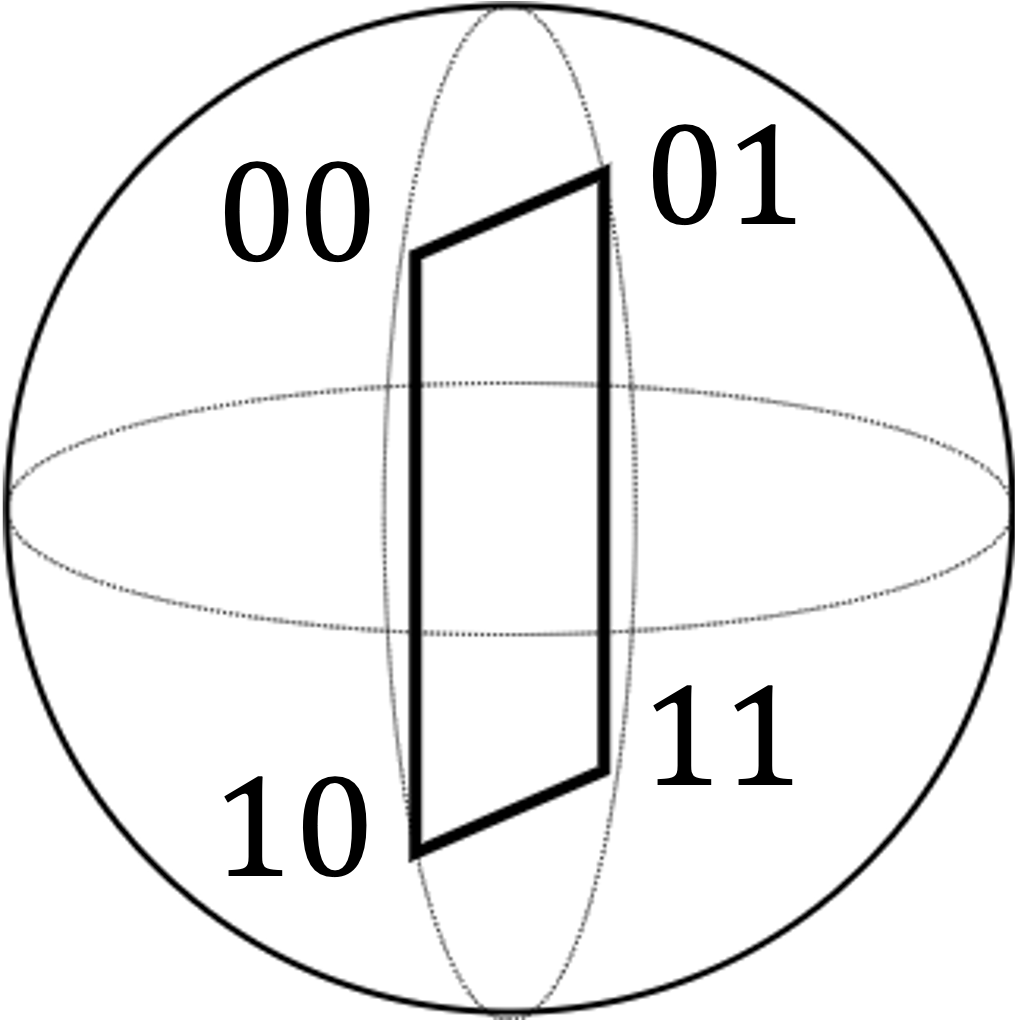}
            \subcaption{$(2,1,0.85)$-QRA coding}
            \label{fig:21bloch}
        \end{minipage} &
        \begin{minipage}[b]{0.3\hsize}
            \centering
            \includegraphics[height=3.8cm]{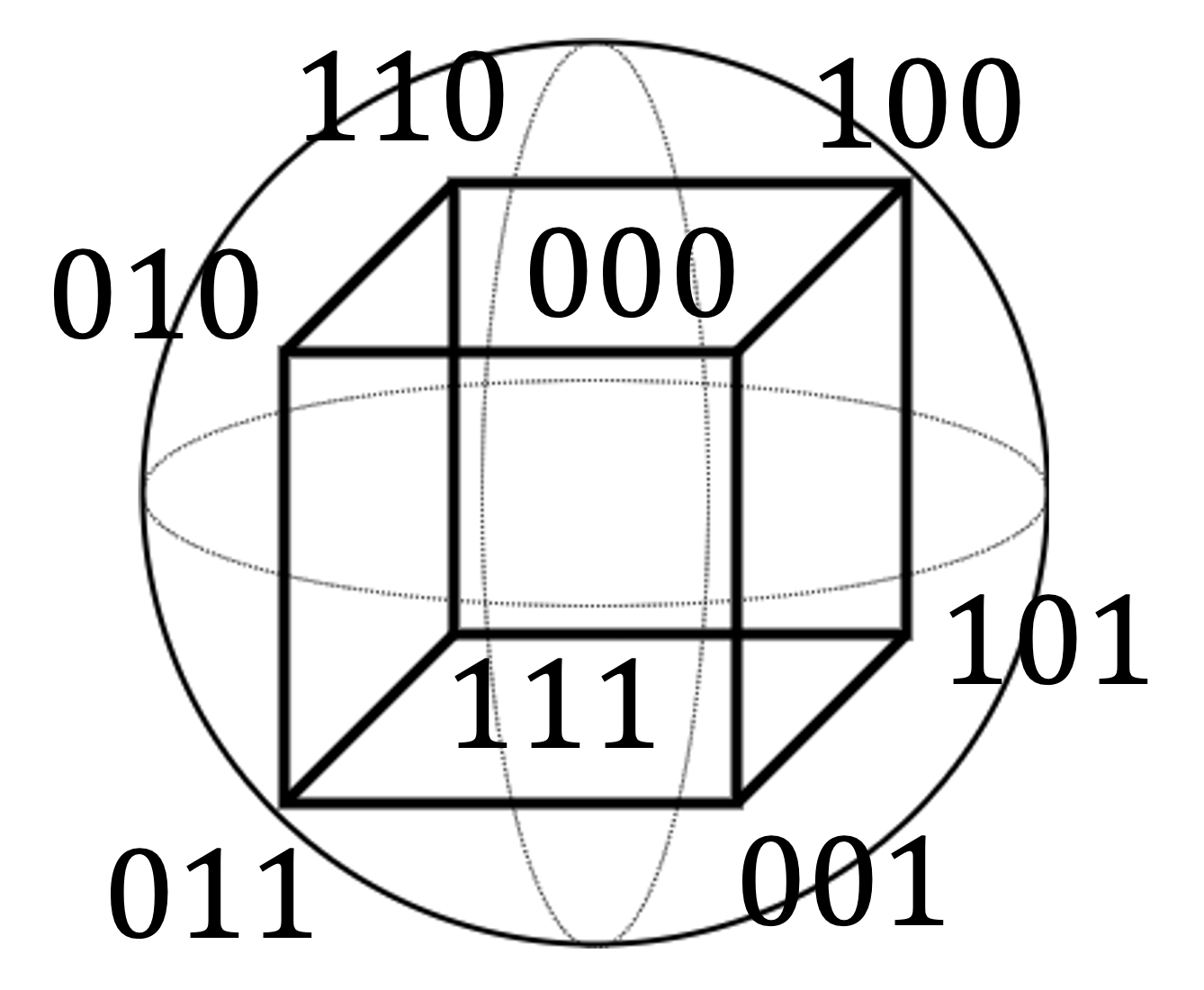}
            \subcaption{$(3,1,0.79)$-QRA coding}
            \label{fig:31bloch}
        \end{minipage}
    \end{tabular}
    \caption{The $(n,1,p)$-QRA coding in Bloch sphere representation}
    \label{fig:qrac_bloch}
\end{figure}

For $n\geq 4$, it is proved that there exists no $(n,1,p)$-QRA coding where $p>\frac{1}{2}$~\cite{hayashi20064}.
For $(n,m,p)$-QRA coding where $m\geq 2$, it is proved that $(n,m,p)$-QRA coding does not exists if $n\geq 4^m$~\cite{hayashi20064}.
It is also proved that $(n,m,p)$-QRA coding exists if and only if $n\leq 4^m-1$~\cite{iwama2007unbounded}.
From here, we sometimes write $(n,m,p)$-QRA codes as $(n,m)$-QRAC for simplicity.
About the success probability of decoding $p$, it is known that $\frac{1}{2}+\frac{1}{\sqrt{2}}\approx0.85$ for $(2,1)$-QRAC and $\frac{1}{2}+\frac{1}{2\sqrt{3}}\approx 0.79$ for $(3,1)$-QRAC are optimal, for example, based on the bound $\frac{1}{2}+\frac{1}{2\sqrt{n}}$ for $(n,1)$-QRAC~\cite{ambainis2002dense}.
For $(m,n)$-QRAC, there is a bound for $p$ known as Nayak bound~\cite{nayak1999optimal}:
\begin{equation}
    m>(1-H(p))n
\end{equation}
where $H(p)$ is the binary entropy function:
\begin{equation}
    H(p):=-p\log_2p-(1-p)\log_2(1-p).
\end{equation}
Recently, the better bound for $p$ than Nayak's one for some pairs of $(m,n)$~\cite{manvcinska2022geometry} is obtained as:
\begin{equation}
    p\leq\frac{1}{2}+\frac{1}{2}\sqrt{\frac{2^{m-1}}{n}}.
    \label{eq:better_bound}
\end{equation}
The above bound gives $p=\frac{1}{2}+\frac{1}{2\sqrt{2}}$ for $(4,2)$-QRAC and $p=\frac{1}{2}+\frac{1}{2\sqrt{3}}$ for $(6,2)$-QRAC.
These bound for $p$ implies the impossibility to make better $(4,2)$- or $(6,2)$-QRAC than just dividing $4$ or $6$ classical bits to the pair of $2$ or $3$ bits and using two $(2,1)$- or $(3,1)$-QRACs for each group of $2$ or $3$ bits independently.
The bound \Cref{eq:better_bound} also implies the optimality of $p=\frac{1}{2}+\frac{1}{\sqrt{6}}$ for $(3,2)$-QRAC~\cite{imamichi2018constructions} obtained by numerical calculation.
We will see this $(3,2)$-QRAC more in detail in \Cref{sec:32qrao} because it is used to extend the original QRAO~\cite{fuller2021approximate} using explained in the next \Cref{sec:qrao}.

\subsection{Quantum Relaxation Based Optimization Algorithms}
\label{sec:qrao}
The following explanation is based on the QRAO paper~\cite{fuller2021approximate}.
We explain the quantum-relaxation based optimization algorithm by using the MaxCut problem formulated as
\begin{equation}
    \max_{\{-1,+1\}^{|V(G)|}}\frac{1}{2}\sum_{e_{i,j}\in E(G)}(1-x_ix_j).
\end{equation}
In the typical quantum-classical hybrid approach using variational methods such as VQE~\cite{peruzzo2014variational} or QAOA~\cite{farhi2014quantum}, each classical binary variable $x_i$ is mapped to $i$-th qubit using the Pauli $Z$ operator. Then the MaxCut problem is reduced to the problem to find the maximum eigenstate of the Hamiltonian:
\begin{equation}
    H=\frac{1}{2}\sum_{e_{i,j}\in E(G)}(I-Z_iZ_j).
    \label{eq:normal_hamiltonian}
\end{equation}
Variational methods such as VQE are used to search for the maximum eigenstate of $H$.
Because $H$ is a diagonal Hamiltonian, it contains the classical states (without superposition or entanglement) as the maximal eigenstates so that the found state in the algorithm can be interpreted directly as the classical solution to the MaxCut problem by just measuring it in the computational basis.

On the other hand, in the quantum-relaxation based optimization algorithms such as QRAO~\cite{fuller2021approximate}, multiple classical bits are encoded into a smaller number of qubits using QRACs explained in \Cref{sec:qrac}.
For example, if we use $(3,1)$-QRAC in \Cref{eq:31qrac}, three classical binary variables $x_1$, $x_2$, and $x_3$ are mapped to a single qubit using the Pauli $X$, $Y$, and $Z$ operators respectively.
Compared with QAOA or VQE, QRAO has the constant-factor space complexity advantage.
We will focus on the QRAO using $(3,1)$-QRAC from here in this section.
The goal is, as well as the typical methods, to reduce the MaxCut problem to the procedure to explore the maximum eigenstate of the Hamiltonian called \textit{relaxed Hamiltonian} $H_{relax}$.
To construct a relaxed Hamiltonian, we make the mapping from classical binary variables into qubits.
First we perform a coloring of the instance graph $G$ by using, for example, LDF (large-degree-first) method~\cite{welsh1967upper} whose time complexity is $O(|V(G)|\log{|V(G)|}+\mathrm{deg}(G)|V(G)|)$ where $\mathrm{deg}(G)$ is the maximum degree of the graph $G$.
After performing the LDF algorithm, the vertices are partitioned into the set $\{V_c\}$ associated with the color $c\in C$.
Let $\mathrm{color}(i)$ be the color of the $i$-th vertex $v_i$.
Then, the following condition holds:
\begin{equation}
    e_{i,j}\in E(G)\implies\mathrm{color}(i)\neq\mathrm{color}(j).
\end{equation}
Next, we associate $\left\lceil\frac{|V_c|}{3}\right\rceil$ qubits for each color $c\in C$.
Now up to three vertices are assigned to a single qubit.
We greedily order these three vertices and assign the Pauli operators $X$, $Y$, and $Z$ respectively.
If we use the $(2,1)$-QRAC, then we associate $\left\lceil\frac{|V_c|}{2}\right\rceil$ qubits for each color and assign the Pauli $X$ and $Z$ for the up to two vertices assigned to the same single qubit instead.
Finally, we obtained a relaxed Hamiltonian instead of the \textit{normal} Hamiltonian in \Cref{eq:normal_hamiltonian} as below:
\begin{equation}
    H_{relax}=\frac{1}{2}\sum_{e_{i,j}\in E(G)}(I-3P_iP_j),
    \label{eq:relax_hamiltonian}
\end{equation}
where $P_i$ is the Pauli operator associated with the vertex $v_i$.
Actually, the typical algorithm using the normal Hamiltonian in \Cref{eq:normal_hamiltonian} can be assumed to be the quantum-relaxation based optimization using $(1,1)$-QRAC defined in \Cref{eq:11qrac}.
We explore the maximum eigenstate of $H_{relax}$ by using variational methods such as VQE.
The relaxed Hamiltonian $H_{relax}$ is no longer diagonal and it contains the non-classical states (with superposition and entanglement) as the maximal eigenstates.
It means that the found eigenstate for the relaxed Hamiltonian cannot be associated with the classical solution directly.
Because of the construction of the Hamiltonian, the found state should be a quantum state that corresponds to the relaxed solution to the MaxCut problem.
A relaxed solution means the solution of the MaxCut problem without the constraint that the solution must be a binary vector.
We denote the found eigenstate in quantum-relaxation based optimization algorithm as $\rho_{relax}$ and called it \textit{relaxed state}.
To retrieve the classical solution for the MaxCut problem, we perform quantum state rounding algorithms.
There are two types of rounding algorithms proposed by Fuller et al.~\cite{fuller2021approximate}.

The first rounding algorithm is \textit{Pauli rounding} which decodes the encoded three classical bits in each qubit by using the POVM defined in \Cref{eq:31povm}.
More precisely, we perform three POVM $E^1,E^2,E^3$ ($X$,$Y$,$Z$ basis measurement) for all qubits with enough shots and calculate the expectation of the $\mathrm{Tr}[M(v_i)\rho_{relax}]$ denoted by $\mathrm{est}_i$ for all vertices $v_i\in V(G)$ where $M$ is an assignment from vertex to Pauli operator. The Pauli operator $X$, $Y$, and $Z$ corresponds to the observables of the POVMs $E_1$, $E_2$, and $E_3$.
After that, we decode the corresponding classical binary value according to $\mathrm{sign}(\mathrm{est}_i)$.
This procedure is equivalent to just measuring the $j$-th qubit with enough shots, taking the majority of the measurement result, and setting it to the rounded value of the corresponding classical bit.

Unfortunately, if the relaxed state is very entangled and cannot be written in the form $\rho_1\otimes\rho_2\otimes\cdots\otimes\rho_n$, there is no guarantee that the Pauli rounding works well because the correlation among qubits is not considered in the Pauli rounding algorithm.
By using the second rounding algorithm, \textit{magic state rounding}, we can avoid the above problem and can obtain the approximation ratio bound for the MaxCut problem.
The idea of the magic state rounding algorithm is to decode three classical variables at once from a single qubit.
Consider the single qubit magic state:
\begin{equation}
    \mu^{\pm}:=\frac{1}{2}\left(I\pm\frac{1}{\sqrt{3}}(X+Y+Z)\right),
\end{equation}
and set
\begin{align}
    \mu^{\pm}_1&:=\mu^{\pm},\\
    \mu^{\pm}_2&:=X\mu^{\pm}X=\frac{1}{2}\left(I\pm\frac{1}{\sqrt{3}}(X-Y-Z)\right),\\
    \mu^{\pm}_3&:=Y\mu^{\pm}Y=\frac{1}{2}\left(I\pm\frac{1}{\sqrt{3}}(-X+Y-Z)\right),\\
    \mu^{\pm}_4&:=Z\mu^{\pm}Z=\frac{1}{2}\left(I\pm\frac{1}{\sqrt{3}}(-X-Y+Z)\right).
\end{align}
In the magic state rounding algorithm, one of the measurement basis $\{\mu^+_i,\mu^-_i\}$ is selected from $i\in[4]$ for each qubit.
After choosing the bases for all qubits, then a relaxed state $\rho_{relax}$ is measured on those bases.
Three classical binary variables are decoded according to the measurement outcome for each qubit.
\Cref{fig:magic} shows the intuition of the magic state rounding algorithm.
Each measurement $\mu^{\pm}_i$ decodes one of the pair of three bits located at opposite angles on the cube (e.g. $000$ or $111$ in the case of $\mu^{\pm}_1$).
By using this simultaneous decoding of the encoded three bits, the magic state rounding algorithm extracts the solution of the MaxCut for every iteration.
The magic state rounding algorithm repeats this procedure enough times and outputs the best solution.
\begin{figure}[tb]
    \begin{tabular}{cccc}
        \begin{minipage}[t]{0.225\hsize}
            \centering
            \includegraphics[height=2.8cm]{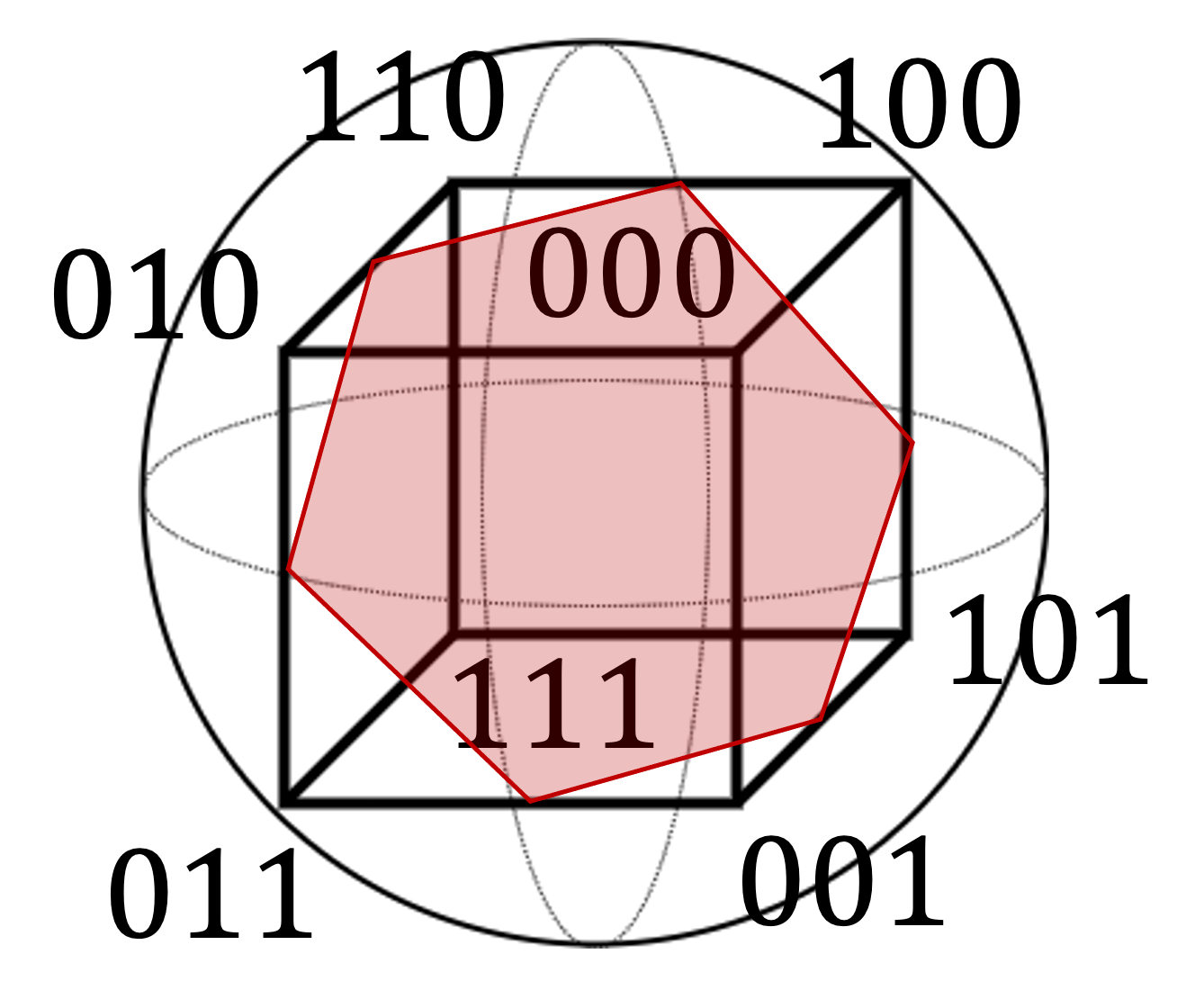}
            \subcaption{$\mu^{\pm}_1$}       
        \end{minipage} &
        \begin{minipage}[t]{0.225\hsize}
            \centering
            \includegraphics[height=2.8cm]{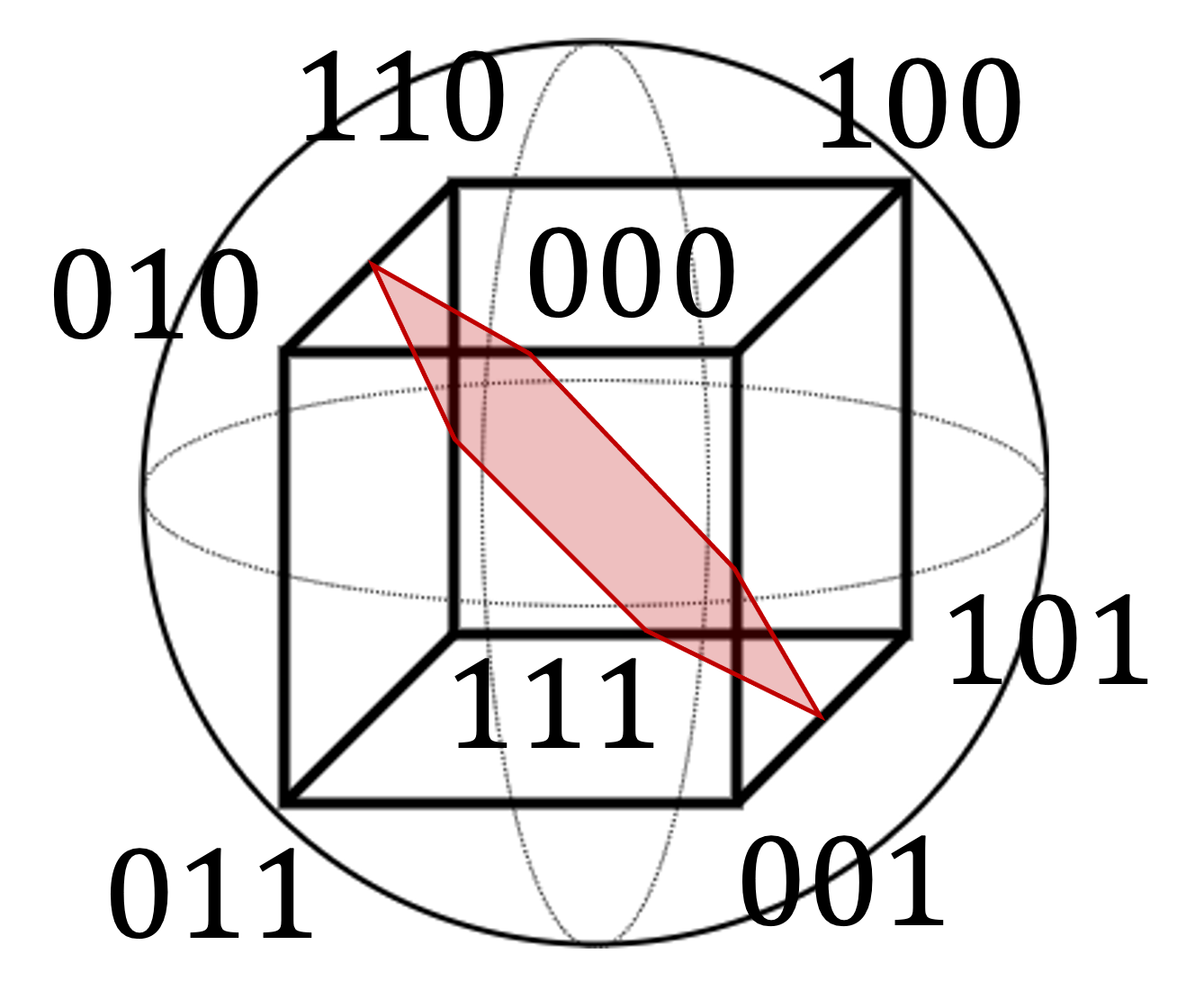}
            \subcaption{$\mu^{\pm}_2$}
        \end{minipage} &
        \begin{minipage}[t]{0.225\hsize}
            \centering
            \includegraphics[height=2.8cm]{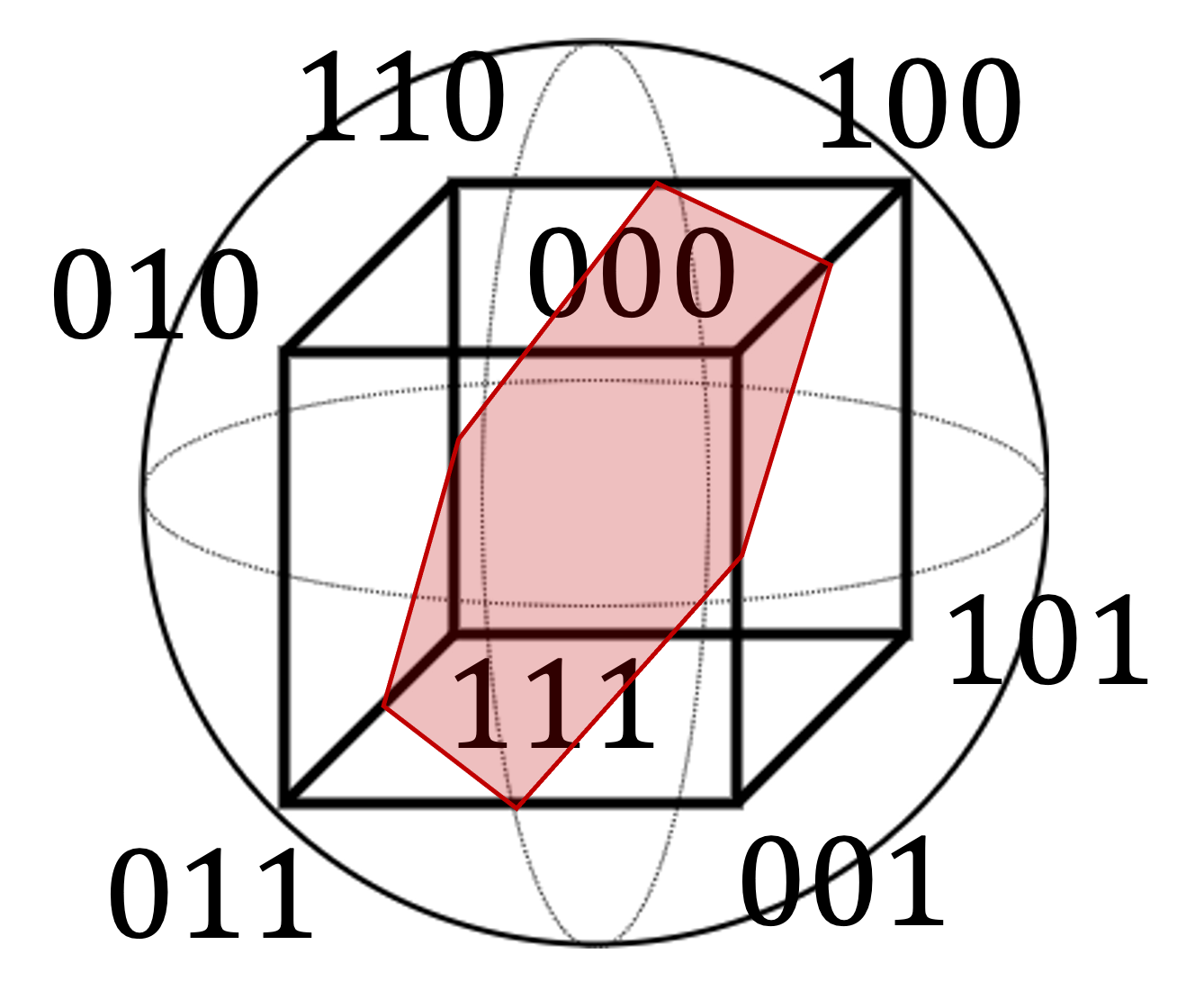}
            \subcaption{$\mu^{\pm}_3$}       
        \end{minipage} &
        \begin{minipage}[t]{0.225\hsize}
            \centering
            \includegraphics[height=2.8cm]{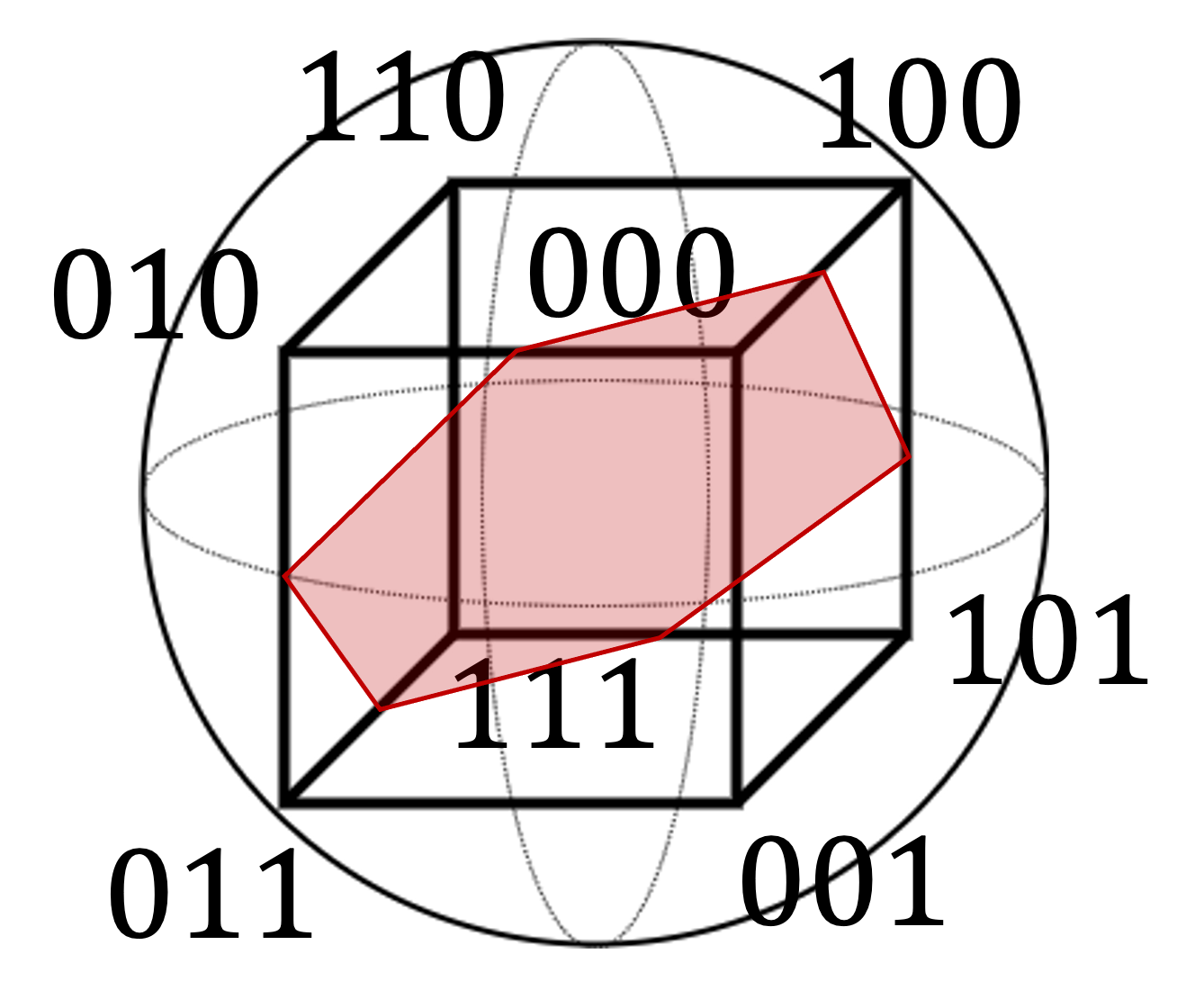}
            \subcaption{$\mu^{\pm}_4$}
        \end{minipage}
    \end{tabular}
    \caption{The intuition of the quantum measurements performed in magic state rounding algorithm}
    \label{fig:magic}
\end{figure}

The approximation ratio bound of QRAO using the magic state rounding algorithm for the MaxCut problem is obtained with the premise that the found relaxed state $\rho_{relax}$ has larger energy than the state associated with the optimum solution, i.e. $\mathrm{Tr}[H_{relax}\rho_{relax}]\geq\mathrm{Tr}[H_{relax}\rho_{opt}]=OPT$
where $\rho_{opt}$ is the quantum state which encodes the optimum solution using $(3,1)$-QRAC and $OPT$ is the optimum value of the instance.
\begin{thm}[\cite{fuller2021approximate}]
    Given access to an oracle $\mathcal{O}_{relax}$ which prepares $\rho_{relax}$ satisfying $\mathrm{Tr}[H_{relax}\rho_{relax}]\geq OPT$, the magic state rounding algorithm solves the MaxCut problem with expected approximation ratio $\mathbb{E}[\gamma]\geq\frac{5}{9}\approx 0.555$.
    \label{thm:31bound}
\end{thm}

We can also prove the approximation ratio bound for the case using $(2,1)$-QRAC (and, of course, the case using $(1,1)$-QRAC).
The measurements used in the magic state rounding algorithm when using $(2,1)$-QRAC are defined like the following:
\begin{equation}
    \xi^{\pm}_1:=\frac{1}{2}\left(I\pm\frac{1}{\sqrt{2}}(X+Z)\right),\xi^{\pm}_2:=\frac{1}{2}\left(I\pm\frac{1}{\sqrt{2}}(X-Z)\right).
\end{equation}
One of the above two measurements is chosen for each qubit.
The expected approximation ratio for the QRAO using $(2,1)$-QRAC is proved to be $\frac{5}{8}=0.625$.
In the case of using $(1,1)$-QRAC, the approximation ratio is obtained as $1.0$.
However, it is meaningless because the existence of the oracle $\mathcal{O}_{relax}$ in the assumption of the proof implies that the oracle can prepare the optimal solution.
It is obvious that given the optimum solution, the approximation ratio is $1.0$.
\Cref{tab:approx_ratios_revisited} summarized the results of the approximation ratios of various quantum-relaxation based optimizers obtained by Fuller et al.~\cite{fuller2021approximate}.
There is a trade-off between the space compression ratio and the approximation ratio.
We will extend QRAO to the case using $(3,2)$-QRAC and prove the approximation ratio bound in \Cref{sec:32qrao}.
\begin{table}[tb]
    \centering
    \caption{The relationship between the approximation ratio for the maximum cut problem and the space compression ratio of quantum-relaxation based optimization algorithms}
    \begin{tabular}{|l||l|l|}
        \hline
        Algorithm & \begin{tabular}{l}space compression \\ratio \end{tabular} & \begin{tabular}{l}approximation \\ratio \end{tabular} \\ \hline
        $(1,1)$-QRAO~\cite{fuller2021approximate} ($\approx$ QAOA~\cite{farhi2014quantum}) & $1.0$ & $(1.0)$ \\ \hline
        $(2,1)$-QRAO~\cite{fuller2021approximate} & $2.0$ & $0.625$ \\ \hline
        $(3,1)$-QRAO~\cite{fuller2021approximate} & $3.0$ & $0.555$ \\ \hline
    \end{tabular}
    \label{tab:approx_ratios_revisited}
\end{table}

\section{Theoretical Extensions of Quantum Relaxations}
\label{sec:32qrao}
\subsection{$(3,2)$-QRA Coding}
\label{subsec:32qrac}
$(3,2)$-QRA coding is one of the quantum random access codes which encodes three classical bits into two qubits.
The concrete formulation of the $(3,2)$-QRAC is obtained in the numerical calculation~\cite{imamichi2018constructions} like the following:
\begin{thm}
    Consider the map from three bits $(x_1,x_2,x_3)\in\{0,1\}^3$ to a two-qubit quantum state $\rho'_{x_1,x_2,x_3}$ defined by the following equations:
    \begin{itemize}
        \item If $b_1\oplus b_2\oplus b_3=0,$\\
            \begin{equation}
                \rho'_{x_1,x_2,x_3}:=\frac{1}{4}I_1I_2+\frac{1}{4}((-1)^{x_1}Z_1I_2+(-1)^{x_2}I_1Z_2+(-1)^{x_3}Z_1Z_2).
                \label{eq:parity0}
            \end{equation}
        \item Else if $b_1\oplus b_2\oplus b_3=1,$
            \begin{equation}
                \begin{split}
                    \rho'_{x_1,x_2,x_3}:=\frac{1}{4}I_1I_2&+(-1)^{x_1}\left(\frac{1}{12}Z_1I_2+\frac{1}{6}X_1X_2+\frac{1}{6}X_1Z_2\right)\\
                    &+(-1)^{x_2}\left(\frac{1}{6}I_1X_2+\frac{1}{12}I_1Z_2+\frac{1}{6}Y_1Y_2\right)\\
                    &+(-1)^{x_3}\left(\frac{1}{12}Z_1Z_2-\frac{1}{6}X_1I_2-\frac{1}{6}Z_1X_2\right)
                    \label{eq:parity1}
                \end{split}
            \end{equation}
    \end{itemize}
    For every pair of $(x_1,x_2,x_3)$, $\rho'_{x_1,x_2,x_3}$ is a pure state.
    Then, this map is a $(3,2,0.908)$-QRA coding with the POVMs (projective measurements, in fact):
    \begin{align}
            F^1&=\left\{\frac{1}{2}I_1I_2\pm\frac{1}{\sqrt{6}}\left(\frac{1}{2}X_1X_2+\frac{1}{2}X_1Z_2+Z_1I_2\right)\right\},\label{eq:povmf1}\\
            F^2&=\left\{\frac{1}{2}I_1I_2\pm\frac{1}{\sqrt{6}}\left(\frac{1}{2}Y_1Y_2+\frac{1}{2}I_1X_2+I_1Z_2\right)\right\},\\
            F^3&=\left\{\frac{1}{2}I_1I_2+\frac{1}{\sqrt{6}}\left(Z_1Z_2-\frac{1}{2}X_1I_2-\frac{1}{2}Z_1X_2\right)\right\}.\label{eq:povmf3}
    \end{align}
    \label{thm:32qrac}
\end{thm}

$(3,2)$-QRAC has two kinds of encoded state form in \Cref{eq:parity0,eq:parity1}, and which to use depends on the parity of the encoded three bits.
It holds that for each parity, four encoded states are orthogonal, i.e. for each $x_1,x_2,x_3\in\{0,1\}^3$ and $x_1',x_2',x_3'\in\{0,1\}^3$ ($(x_1,x_2,x_3)\neq(x_1',x_2',x_3')$) satisfying $x_1\oplus x_2\oplus x_3=x_1'\oplus x_2'\oplus x_3'$,
\begin{equation}
    \left<\psi'(x_1,x_2,x_3)|\psi'(x_1',x_2',x_3')\right>=0.
    \label{eq:orthogonal}
\end{equation}
It implies that if we know the parity of the encoded classical bits in advance, we can decode the encoded three bits by using the 4-outcome quantum measurement.
This characteristic is used when we formulate the rounding algorithm corresponding to the magic state rounding algorithm of the quantum relaxation using $(3,1)$- or $(2,1)$-QRACs.
The POVMs in \Crefrange{eq:povmf1}{eq:povmf3} are used when we'd like to decode the encoded bits one by one (e.g. the Pauli rounding algorithm).
The success probability of the decoding is $\frac{1}{2}+\frac{1}{\sqrt{6}}\approx 0.908$, and it is proved to be optimal by using the bound in \Cref{eq:better_bound}~\cite{manvcinska2022geometry}.
While the space compression ratio of $(3,2)$-QRAC is less than $(3,1)$- or $(2,1)$-QRACs, the success probability of decryption is better than theirs.
We will see in the next section that the same kind of trade-off with the space compression ratio also holds for the approximation ratio of the quantum relaxation using them.

\subsection{Quantum Relaxation Using $(3,2)$-QRAC}
As we see in \Cref{sec:qrao}, we have to extend the problem Hamiltonian $H_{relax}$ for $(3,2)$-QRAC.
Fortunately, we can achieve this step by just substituting the Pauli $X$, $Y$, and $Z$ operators that appeared in $H_{relax}$ by the two-qubit operators $X'$, $Y'$, and $Z'$ respectively and changing the coefficient of the 2-local Pauli operators to $6$.
The definitions of $X'$, $Y'$, and $Z'$ are given in the following equations:
\begin{align}
    X'&:=\frac{1}{\sqrt{6}}\left(\frac{1}{2}X_1X_2+\frac{1}{2}X_1Z_2+Z_1I_2\right),\\
    Y'&:=\frac{1}{\sqrt{6}}\left(\frac{1}{2}I_1X_2+I_1Z_2+\frac{1}{2}Y_1Y_2\right),\\
    Z'&:=\frac{1}{\sqrt{6}}\left(Z_1Z_2-\frac{1}{2}X_1I_2-\frac{1}{2}Z_1X_2\right).
\end{align}
These operators satisfy similar conditions as Pauli operators.
For example,
\begin{equation}
    \mathrm{Tr}[X']=\mathrm{Tr}[Y']=\mathrm{Tr}[Z']=0,
\end{equation}
and for $P'_i,P'_j\in\{X',Y',Z'\}$,
\begin{equation}
    \mathrm{Tr}[P'_i\cdot P'_j]=\delta_{i,j}.
\end{equation}
However, these operators do not satisfy the anti-commutation rule while the Pauli operators satisfy it.
The algorithms are almost the same as QRAO using $(3,1)$-QRAC.
The first step of the algorithm is to color the vertices of the graph.
After that, we make pairs of two qubits and assign a single pair to up to $3$ vertices for which the same color is assigned in graph coloring.
For each vertex assigned to the same pair of two qubits, $X'$, $Y'$, and $Z'
$ is assigned in order instead of the Pauli $X$, $Y$, and $Z$ operators.
Now, all vertices of the graph are associated with one of the operators $X'$, $Y'$, and $Z'$ acting on the same or distinct pair of two qubits.
Intuitively, it can be interpreted as if one qubit in the case of $(3,1)$-QRAC is made redundant by two qubits, and the Pauli operators are replaced with the operators $\{X',Y',Z'\}$.
Then, the problem Hamiltonian of the quantum relaxation using $(3,2)$-QRAC denoted by $H'_{relax}$ is defined like the following:
\begin{equation}
    H'_{relax}:=\frac{1}{2}\sum_{e_{i,j}\in E(G)}(I-6P'_iP'_j)
\end{equation}
where $P'_i$ is one of the operators $\{X',Y',Z'\}$ associated with the vertex $v_i$.
The coefficient of $P'_iP'_j$ is $-1$ because of the relation:
\begin{equation}
    \begin{split}
        \mathrm{Tr}[X'\cdot\rho'(x_1,x_2,x_3)]&=\frac{1}{\sqrt{6}}(-1)^{x_1},\\
        \mathrm{Tr}[Y'\cdot\rho'(x_1,x_2,x_3)]&=\frac{1}{\sqrt{6}}(-1)^{x_2},\\
        \mathrm{Tr}[Z'\cdot\rho'(x_1,x_2,x_3)]&=\frac{1}{\sqrt{6}}(-1)^{x_3}.
    \end{split}
    \label{eq:cor32}
\end{equation}
The next step is to find a maximum eigenstate of the relaxed Hamiltonian $H'_{relax}$ by variational methods such as VQE.
Once we obtained the quantum states corresponding to the relaxed solution to the MaxCut problem, the quantum state rounding algorithm is performed to extract the classical solution.

By using the POVMs in \Crefrange{eq:povmf1}{eq:povmf3}, we can define the rounding algorithm which decodes the encoded bits one by one like the Pauli rounding algorithm of QRAO.
We name the algorithm \textit{individual rounding} and define it like the following.
We perform the POVMs $F^1,F^2,F^3$ for all qubits with enough shots and calculate the expectation of the $\mathrm{Tr}[M'(v_i)\rho'_{relax}]$ denoted by $est'_i$ for all vertices $v_i\in V(G)$ where $M'$ is an assignment from vertex to the operators $\{X',Y',Z'\}$.
This operation can be implemented by making a circuit that maps $\ket{00}$ to $F^i_0$ and $\ket{01}$ to $F^i_1$, taking a conjugate of the circuit, performing the circuit and measuring the second qubit.
After that, we decode the corresponding classical binary value according to $\mathrm{sign}(est'_i)$.
The whole procedure is described in Algorithm 1.
\begin{algorithm}[tb]
    \SetKwInOut{Input}{Input}\SetKwInOut{Output}{Output}\SetAlgoNoLine
    \caption{Individual rounding algorithm for the quantum-relaxation based optimization algorithm using $(3,2)$-QRAC}
    \label{alg:pauli32}
    \Input{An oracle $\mathcal{O'}_{relax}$ which prepares relaxed state $\rho'_{relax}$; Number of measurement shots $S'$; An assignment $M'$ from vertex to the operators $\{X',Y',Z'\}$.}
    \Output{Approximate solution $x\in\{0,1\}^{|V(G)|}$}
    Initialize approximate solution $x=(1,1,...,1)$.\\
    Prepare $\rho'_{relax}$ using $\mathcal{O'}_{relax}$.\\
    Measure each qubit by the POVMs $F^1$ $F^2$, and $F^3$ with $S'$ shots respectively.\\
    Calculate the estimation $\mathrm{est'}_i$ of the value $Tr\left(\rho'_{relax}\cdot M'(v_i)\right)$ for each $v_i\in V(G)$.\\
    \For{$i\in [|V(G)|]$}{
        \eIf{$(\mathrm{est'}_i=0)$}{
            Assign the value to $x_i$ uniformly at random.
        }{
            Assign the value to $x_i$ according to $\mathrm{sign}(\mathrm{est'}_i)$.
        }
    }
    \Return{$x$}
\end{algorithm}

On the other hand, to obtain the approximation ratio bound, we need the other rounding algorithm which decodes the configuration of the graph cut by one-shot measurement like the magic state rounding algorithm of QRAO because the Pauli rounding type algorithms do not take the correlation between qubits into account.
The key to constructing the rounding algorithm for approximation ratio is to design the quantum measurement which decodes encoded three bits for each qubit at once.
We name the algorithm \textit{simultaneous rounding} and define it like the following.
In the case of $(3,1)$- or $(2,1)$-QRACs, decoding was performed for each pair of two bit-inverted relationships by using the magic state basis measurements.
In the case of $(3,2)$-QRAC, the measurement performed is a two-qubits measurement.
There will be up to four different measurement results meaning that up to four different bit patterns can be decoded simultaneously.
As we mentioned in \Cref{subsec:32qrac}, if we know the parity of the encoded bits, then we can decode the encoded three bits by using the 4-outcome quantum measurement defined below up to the parity $0$ or $1$.
\begin{equation}
    \{\rho'_{x_1,x_2,x_3}\}_{x_1\oplus x_2\oplus x_3=0},\ \mathrm{or}\ \{\rho'_{x_1,x_2,x_3}\}_{x_1\oplus x_2\oplus x_3=1}.
    \label{eq:magic32}
\end{equation}
These measurements are rank-1 projective measurements:
\begin{lem}
    The measurements in \Cref{eq:magic32} are rank-1 projective measurements.
    \label{lem:proj32}
\end{lem}

In the simultaneous rounding algorithm, one of the parity is chosen randomly for each qubit, and one of the corresponding measurements in \Cref{eq:magic32} is performed to the relaxed state $\rho'_{relax}$.
These measurements are performed for all qubits at once and decode one solution to the MaxCut problem.
To implement the above measurement, for example in the case that the parity is $0$, we apply the following unitary operation:
\begin{equation*}
    \ket{00}\bra{\psi'(0,0,0)}+\ket{01}\bra{\psi'(0,1,1)}+\ket{10}\bra{\psi'(1,0,1)}+\ket{11}\bra{\psi'(1,1,0)},
\end{equation*}
measure the state on the computational basis, and decode the bits according to the two bits measurement results like the following:
\begin{equation}
    \begin{split}
        &00\mapsto 000\\
        &01\mapsto 011\\
        &10\mapsto 101\\
        &11\mapsto 110
    \end{split}
    \label{eq:assign32}
\end{equation}
By repeating this procedure sufficient times and taking the best solution, the simultaneous rounding algorithm for the quantum relaxation using $(3,2)$-QRAC finds a classical solution. The whole procedure is described in Algorithm 2.
\begin{algorithm}[tb]
    \SetKwInOut{Input}{Input}\SetKwInOut{Output}{Output}\SetAlgoNoLine
    \caption{Simultaneous rounding algorithm for the quantum-relaxation based optimization algorithm using $(3,2)$-QRAC}
    \label{alg:magic32}
    \Input{An oracle $\mathcal{O'}_{relax}$ which prepares relaxed state $\rho'_{relax}$; Number of measurement shots $S'$.}
    \Output{Approximate solution $x\in\{0,1\}^{|V(G)|}$}
    Initialize approximate solution $x=(1,1,...,1)$.\\
    \For{$s'\in[S']$}{
        Prepare $\rho'_{relax}$ using $\mathcal{O'}_{relax}$.\\
        Randomly and independently choose the parity $p\in\{0,1\}$ for each qubit.\\
        Measure $\rho'_{relax}$ by $\{\rho'_{x_1,x_2,x_3}\}_{x_1\oplus x_2\oplus x_3=p}$ and assign the binary variables according to the measurement result and basis like \Cref{eq:assign32} for each qubit.\\
        Let the resulting solution be $x'$.\\
        Let $\mathrm{cut}(x)$ be the cut value of $x$.\\
        \If{$\mathrm{cut}(x)<\mathrm{cut}(x')$}{
            $x\leftarrow x'$
        }
    }
    \Return{x} 
\end{algorithm}
The quantum relaxation using $(3,2)$-QRAC and the simultaneous rounding algorithm described above yields the expected approximation ratio bound for the MaxCut problem.
In the next section, we prove the approximation ratio to be $0.722$.

\subsection{Proof of the Approximation Ratio}
In the proof of the approximation ratio for the quantum relaxation using $(3,1)$- or $(2,1)$-QRACs, the quantum measurement performed in the magic state rounding algorithm is equivalent in expectation to the single qubit depolarizing channel.
Then, by using the self-adjointness of the single qubit depolarizing channel and the effect of the channel on the Pauli operators the theoretical bound was obtained.
For the case of $(3,2)$-QRAC, we use a similar discussion to prove the approximation ratio bound.
Let the operation "measuring a two qubits state on a basis $\{\rho'_{x_1,x_2,x_3}\}_{x_1\oplus x_2\oplus x_3=p}$ where $p$ is a randomly chosen parity $0$ or $1$" be $\mathcal{M'}$ that is a measurement performed in the simultaneous rounding algorithm in Algorithm 2.
Let us define the expectation of the measurement $\mathcal{M'}$ to be $\Phi'$ like the following:
\begin{align}
    \Phi'(\tau):&=\mathbb{E}[\mathcal{M'}(\tau)]\\
    &=\sum_{p\in\{0,1\}}\sum_{\substack{x_1,x_2,x_3:\\x_1\oplus x_2\oplus x_3=p}}\frac{1}{2}\cdot\mathrm{Tr}[\rho'_{x_1,x_2,x_3}\tau]\cdot\frac{\sqrt{\rho'_{x_1,x_2,x_3}}\tau\sqrt{\rho'_{x_1,x_2,x_3}}^{\dag}}{\mathrm{Tr}[\rho'_{x_1,x_2,x_3}\tau]}\\
    &=\frac{1}{2}\sum_{p\in\{0,1\}}\sum_{\substack{x_1,x_2,x_3:\\x_1\oplus x_2\oplus x_3=p}}\mathrm{Tr}[\rho'_{x_1,x_2,x_3}\tau]\cdot\rho'_{x_1,x_2,x_3}.\label{eq:other_magic32}
\end{align}
The third equation holds from Lemma 7 and the fact that each $\rho'_{x_1,x_2,x_3}$ is Hermitian.
Unfortunately, the above quantum operation $\Phi'$ is not a depolarizing channel.
However, it is enough for us to have some preferable properties of $\Phi'$ to the operators $X'$, $Y'$, and $Z'$, and the self-adjointness of $\Phi'$ as shown in the following lemmas.
\begin{lem}
    It holds that
    \begin{align*}
        \Phi'(I)&=I,\\
        \Phi'(P')&=\frac{2}{3}P'\ (\forall P'\in\{X',Y',Z'\}).
    \end{align*}
    \label{lem:shrink32}
\end{lem}

\begin{lem}
    The quantum operation $\Phi'$ is self-adjoint with respect to the inner-product $\left<A,B\right>=\mathrm{Tr}[A\cdot B]$, i.e. for any operator $\phi$ and $\rho$,
    \begin{equation*}
        \mathrm{Tr}[\varphi\cdot\Phi'(\tau)]=\mathrm{Tr}[\Phi'(\varphi)\cdot \tau].
    \end{equation*}
    \label{lem:self_adjoint32}
\end{lem}

By using these facts, we obtained the approximation ratio bound of the quantum-relaxation based optimizer using $(3,2)$-QRAC for the MaxCut problem under the premise that the found relaxed state's energy is larger than the energy of the quantum state associated with the optimum solution:
\begin{thm}
    Consider an oracle $\mathcal{O'}_{relax}$ which prepares the relaxed state $\rho'_{relax}$ for the quantum relaxation using $(3,2)$-QRAC satisfying the condition $\mathrm{Tr}[H'_{relax}\rho'_{relax}]\geq OPT$ where $OPT$ is the optimum value.
    Given access to $\mathcal{O'}_{relax}$, the simultaneous rounding algorithm solves the MaxCut problem with an expected approximation ratio $\mathbb{E}[\gamma]\geq\frac{13}{18}\approx 0.722$.
\end{thm}
\begin{proof}
    Let $n$ be the number of qubits involved in the algorithm.
    If $n$ is an odd number, we consider the dummy vertices of the graph to make $n$ even.
    By definition,
    \begin{align*}
        \mathbb{E}[\gamma]&=\mathbb{E}\left[\frac{\mathrm{Tr}[H'_{relax}\mathcal{M'}^{\otimes \frac{n}{2}}(\rho'_{relax})]}{OPT}\right]\\
        &=\frac{1}{OPT}\cdot\left(\frac{|E(G)|}{2}+\mathrm{Tr}\left[\left(H'_{relax}-\frac{|E(G)|}{2}I^{\otimes n}\right)\cdot\Phi'^{\otimes\frac{n}{2}}(\rho'_{relax})\right]\right).
    \end{align*}
    By using the self-adjointness of the operation $\Phi'$ (Lemma 9),
    \begin{equation*}
        \mathbb{E}[\gamma]=\frac{1}{OPT}\cdot\left(\frac{|E(G)|}{2}+\mathrm{Tr}\left[\Phi'^{\otimes\frac{n}{2}}\left(H'_{relax}-\frac{|E(G)|}{2}I^{\otimes n}\right)\cdot\rho'_{relax}\right]\right).
    \end{equation*}
    The operator $H'_{relax}-\frac{|E(G)|}{2}I^{\otimes n}$ is a weighted sum of $P'Q'$ where $P',Q'\in\{X',Y',Z'\}$ and acting on a distinct pair of two qubits. By Lemma 8,
    \begin{equation*}
        \mathbb{E}[\gamma]=\frac{1}{OPT}\cdot\left(\frac{|E(G)|}{2}+\left(\frac{2}{3}\right)^2\cdot\mathrm{Tr}\left[\left(H'_{relax}-\frac{|E(G)|}{2}I^{\otimes n}\right)\cdot\rho'_{relax}\right]\right).
    \end{equation*}
    From the assumption $\mathrm{Tr}[H'_{relax}\rho'_{relax}]\geq OPT$,
    \begin{equation*}
        \mathbb{E}[\gamma]\geq\frac{\frac{|E(G)|}{2}+\frac{4}{9}\cdot\left(OPT-\frac{|E(G)|}{2}\right)}{\frac{|E(G)|}{2}+\left(OPT-\frac{|E(G)|}{2}\right)}.
    \end{equation*}
    Because $0\leq OPT-\frac{|E(G)|}{2}\leq\frac{|E(G)|}{2}$,
    \begin{equation*}
        \mathbb{E}[\gamma]\geq\frac{\frac{|E(G)|}{2}+\frac{4}{9}\cdot\frac{|E(G)|}{2}}{\frac{|E(G)|}{2}+\frac{|E(G)|}{2}}=\frac{1+\frac{4}{9}}{1+1}=\frac{13}{18}\approx 0.722.
    \end{equation*}
\end{proof}

\Cref{tab:approx_ratios_rerevisited} shows our result for the quantum relaxation using $(3,2)$-QRAC (denoted by $(3,2)$-QRAO) and the previous results by Fuller et al. for QRAOs.
Our result is consistent with the trade-off between the bit-to-qubit compression ratio and the approximability of quantum-relaxation based optimizers.
\begin{table}[tb]
    \centering
    \caption{The relationship between the approximation ratio for the maximum cut problem and the space compression ratio of quantum-relaxation based optimization algorithms}
    \begin{tabular}{|l||l|l|}
        \hline
        Algorithm & \begin{tabular}{l}space compression \\ratio \end{tabular} & \begin{tabular}{l}approximation \\ratio \end{tabular} \\ \hline
        $(1,1)$-QRAO~\cite{fuller2021approximate} ($\approx$ QAOA~\cite{farhi2014quantum}) & $1.0$ & $(1.0)$ \\ \hline
        $(2,1)$-QRAO~\cite{fuller2021approximate} & $2.0$ & $0.625$ \\ \hline
        $(3,1)$-QRAO~\cite{fuller2021approximate} & $3.0$ & $0.555$ \\ \hline
        $(3,2)$-QRAO & $1.5$ & $0.722$ (our result) \\
        \hline
    \end{tabular}
    \label{tab:approx_ratios_rerevisited}
\end{table}

\subsection{Space Compression Ratio Preserving Quantum Relaxation}
Though QRAO using $(3,1)$- or $(2,1)$-QRACs have a constant-factor space advantage against typical quantum optimizers, the bit-to-qubit compression ratio becomes lower as the density of the graph instance increases.
This is because there is a constraint that the endpoints of each edge must be associated with different qubits.
For example, if the graph instance is the complete graph, then the number of qubits needed to run QRAO is the same as the number of vertices.
In such cases, the quantum-relaxation based optimizer has no space advantage against standard QAOA and VQE algorithms.
In this section, we propose new types of encoding which encode up to two classical bits into a single qubit by using $(3,1)$-QRAC.
Concretely, we encode the parity of the two bits to the third bit's position in $(3,1)$-QRAC formulation like the following:
\begin{equation}
    (x_1,x_2)\mapsto\tilde{\rho}_{x_1,x_2}:=\frac{1}{2}\left(I+\frac{1}{\sqrt{3}}((-1)^{x_1}X+(-1)^{x_2}Y+(-1)^{x_1\oplus x_2}Z)\right).
    \label{eq:tetra}
\end{equation}
\Cref{fig:tetrahedron} shows the Bloch sphere representation of $(3,1)$-QRAC and the encoding of \Cref{eq:tetra}.
\Cref{eq:tetra} encodes the two classical bits into one of the four vertices of the tetrahedron visualized in \Cref{fig:tetra_bloch}.
These four vertices correspond to the four of eight vertices of the cube in the case of $(3,1)$-QRAC in \Cref{fig:31bloch_revisit}.
\begin{figure}[tb]
    \begin{tabular}{cc}
        \begin{minipage}[b]{0.45\hsize}
            \centering
            \includegraphics[height=5cm]{figures/31qrac.png}
            \subcaption{$(3,1)$-QRAC}       
            \label{fig:31bloch_revisit}
        \end{minipage} &
        \begin{minipage}[b]{0.45\hsize}
            \centering
            \includegraphics[height=5cm]{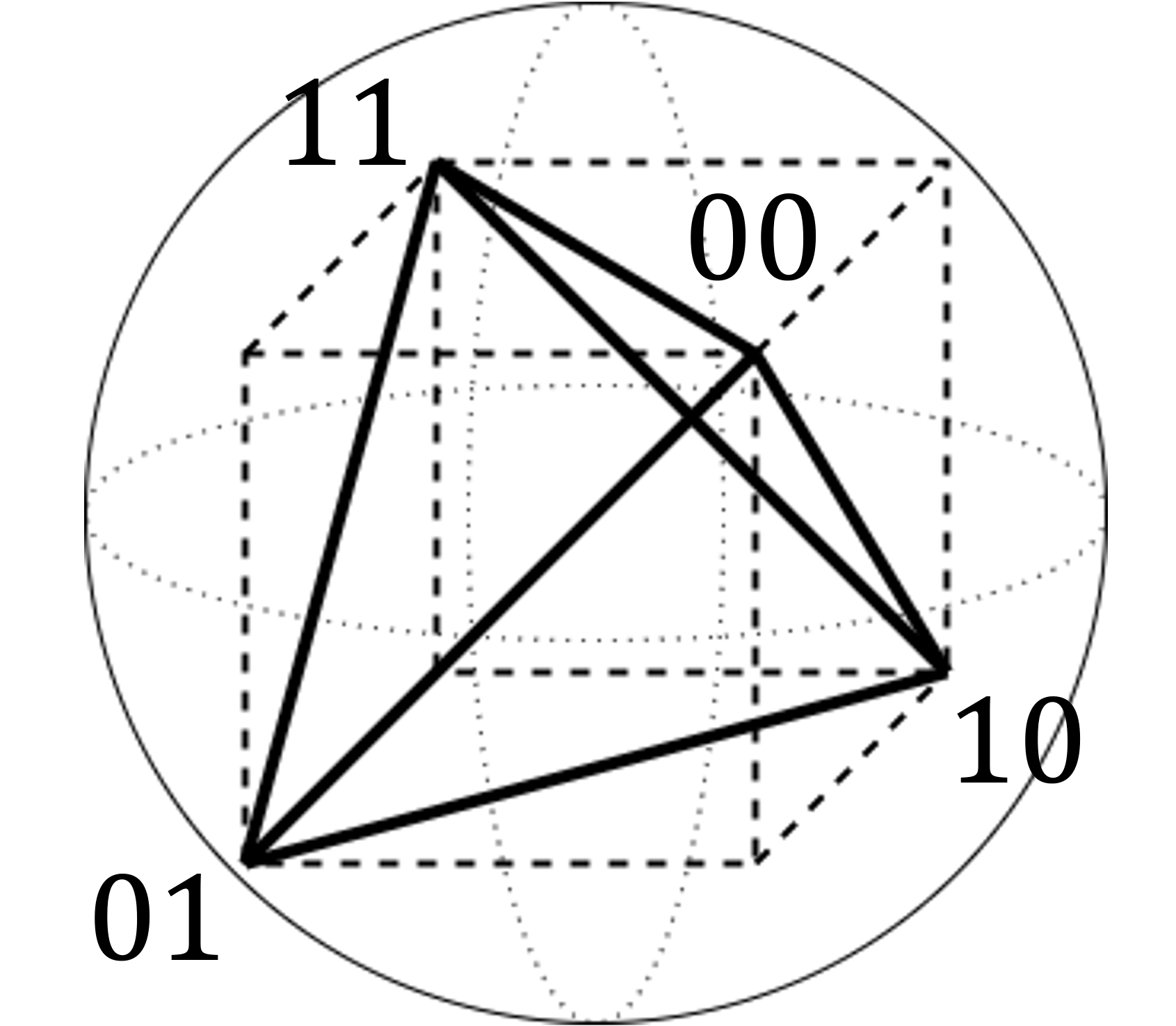}
            \subcaption{Encoding of \Cref{eq:tetra}}
            \label{fig:tetra_bloch}
        \end{minipage}
    \end{tabular}
    \centering
    \caption{The Bloch sphere representation of $(3,1)$-QRAC and the encoding of \Cref{eq:tetra}}
    \label{fig:tetrahedron}
\end{figure}

Let us formulate the quantum relaxation based on the encoding in \Cref{eq:tetra}.
In the QRAO by Fuller et al., a graph coloring algorithm is performed as preprocessing to satisfy the constraint that the endpoints of each edge must be assigned to different qubits.
On the contrary, in our new space compression ratio preserving quantum relaxation, such preprocessing is unnecessary.
We just partition the vertices into $\frac{|V(G)|}{2}$ pairs of two vertices and assign the Pauli $X$ or $Y$ to the two vertices respectively.
Here, w.l.o.g. we assume that $|V(G)|$ is even because otherwise we just add one dummy vertices to make $|V(G)|$ even.
Then, we construct the relaxed Hamiltonian from the instance graph.
The definition of the relaxed Hamiltonian for space compression ratio preserving quantum relaxation is almost the same as the one for $(3,1)$-QRAO.
For each edge $(i,j)\in E(G)$, if the endpoints of it are assigned to different qubits, we encode the edge as the term $P_iP_j$ where $P_i\in\{X,Y\}$ are the Pauli operators associated with the vertex of index $i$.
If the endpoints of the edge are assigned to the same qubit, we use the Pauli $Z$ operator acting on the qubit.
Let $\mathrm{Q\_idx}(i)$ be the index of the qubit associated with the $i$-th vertex.
Formally, the relaxed Hamiltonian for our quantum relaxation $\tilde{H}_{relax}$ is defined like the following:
\begin{equation}
    \tilde{H}_{relax}:=\frac{1}{2}\sum_{e:=(i,j)\in E(G)}(I-O_e)
    \label{eq:tetra_hamiltonian}
\end{equation}
where
\begin{equation}
    O_e:=\begin{cases}3P_iP_j & \mathrm{if}\ \mathrm{Q\_idx}(i)\neq\mathrm{Q\_idx}(j),\\ \sqrt{3}Z_k & \mathrm{if}\ \mathrm{Q\_idx}(i)=\mathrm{Q\_idx}(j)=k.\end{cases}
    \label{eq:oe}
\end{equation}
We note that $P_i$ and $P_j$ in \Cref{eq:oe} are $X$ or $Y$ acting on the different qubits $\mathrm{Q\_idx}(i)$ and $\mathrm{Q\_idx}(j)$.
The coefficient $\sqrt{3}$ in \Cref{eq:oe} comes from the relation:
\begin{equation}
    \mathrm{Tr}[\tilde{\rho}_{x_1,x_2}Z]=\frac{1}{\sqrt{3}}(-1)^{x_1\oplus x_2}.
\end{equation}
As well as the other quantum relaxations, we explore the maximum eigenstate of $\tilde{H}_{relax}$ and find the candidate relaxed state $\tilde{\rho}_{relax}$.
\begin{figure}[tb]
    \begin{tabular}{cccc}
        \begin{minipage}[t]{0.225\hsize}
            \centering
            \includegraphics[height=2.8cm]{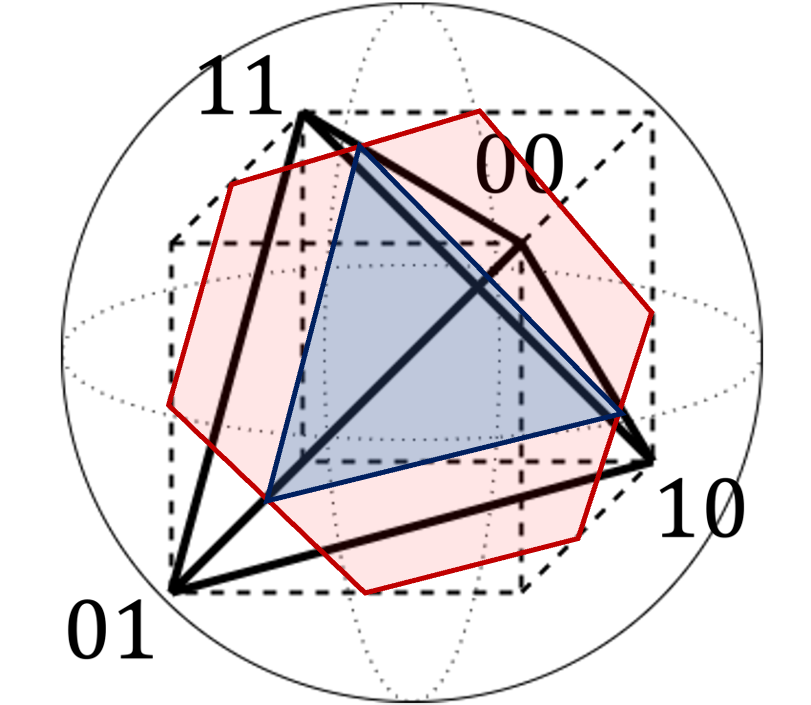}
            \subcaption{$\mu^{\pm}_1$}       
        \end{minipage} &
        \begin{minipage}[t]{0.225\hsize}
            \centering
            \includegraphics[height=2.8cm]{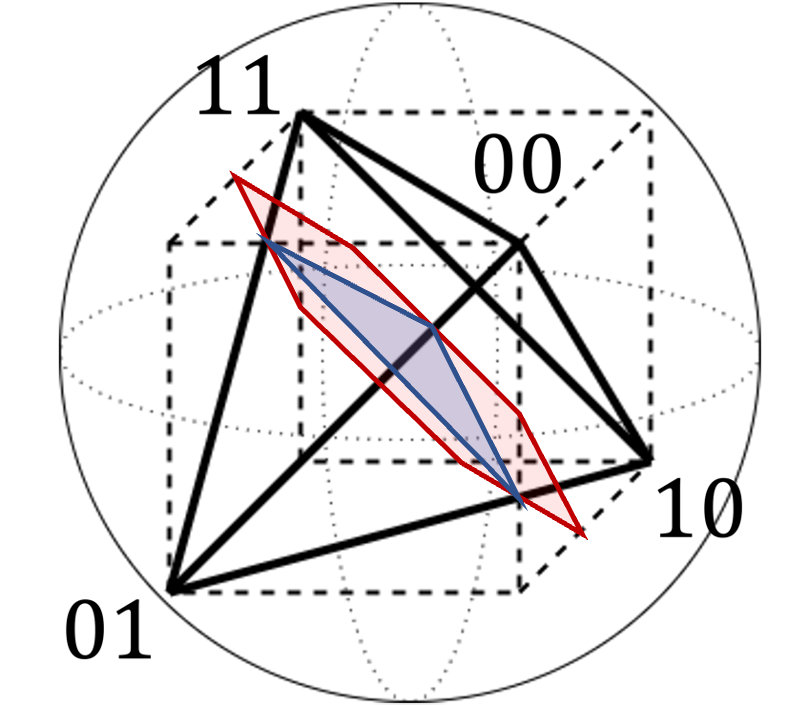}
            \subcaption{$\mu^{\pm}_2$}
        \end{minipage} &
        \begin{minipage}[t]{0.225\hsize}
            \centering
            \includegraphics[height=2.8cm]{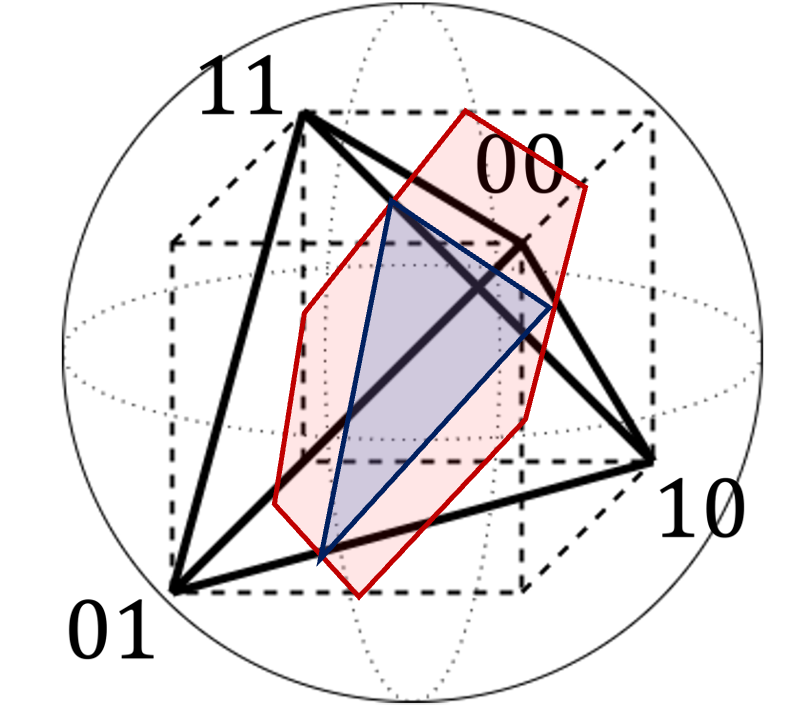}
            \subcaption{$\mu^{\pm}_3$}       
        \end{minipage} &
        \begin{minipage}[t]{0.225\hsize}
            \centering
            \includegraphics[height=2.8cm]{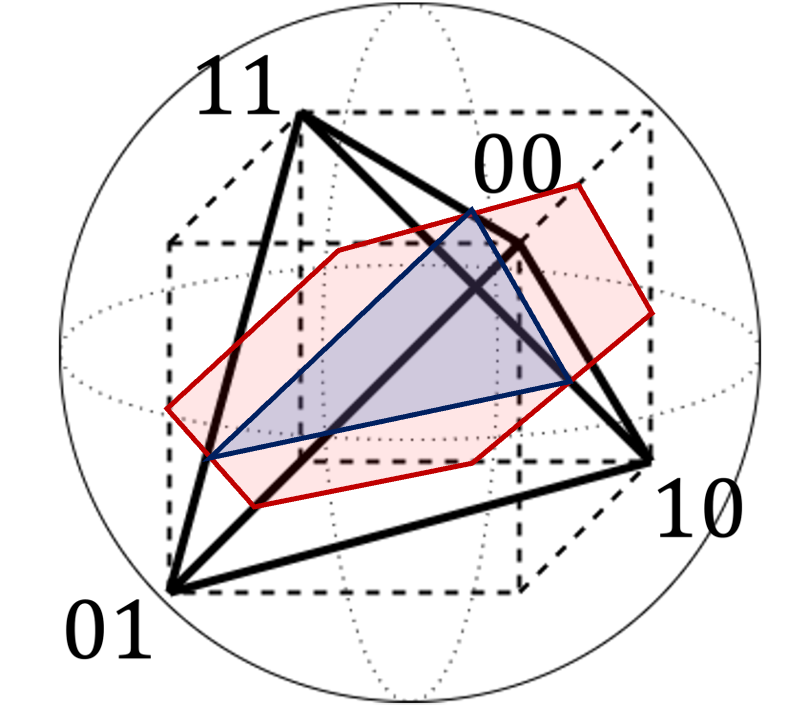}
            \subcaption{$\mu^{\pm}_4$}
        \end{minipage}
    \end{tabular}
    \caption{The intuition of the quantum measurements performed in magic state rounding algorithm performed in space compression ratio preserving quantum relaxation}
    \label{fig:magic_tetra}
\end{figure}

The next step is to define the quantum state rounding algorithm.
The Pauli rounding is the same as that for $(3,1)$-QRAO but disregards the third encoded bit.
The magic state rounding algorithm for our quantum relaxation is also the same as that for $(3,1)$-QRAO but the decoding rule is different.
By the magic bases $\mu^{\pm}_1,\mu^{\pm}_2,\mu^{\pm}_3,\mu^{\pm}_4$, the four encoded patterns $00,01,10,11$ are divided into $2$ groups containing $1$ and $3$ patterns.
The intuition of the magic state measurement is described in \Cref{fig:magic_tetra}.
The red-lined hexagon in the Bloch sphere represents the plane corresponding to the magic state measurement, and the blue-lined triangle represents the intersection of the hexagon and the tetrahedron.
For instance, $\mu^{\pm}_1$ divides the patterns into $\{00\}$ and $\{01,10,11\}$.
If we measure the encoded state $\tilde{\rho}_{0,0}$ in this basis, then the measurement result is always $0$.
If we measure the other three encoded states, then the measurement result is always $1$.
Conversely, if we do not know the encoded two bits, and the measurement outcome of $\mu^{\pm}_1$ is $0$, then the encoded bits are decided to be $00$.
Otherwise, the probabilities that the encoded bits are $01$, $10$, and $11$ are the same $\left(\frac{1}{3}\right)$.
From the above discussions, we define the decoding rule for $\mu^{\pm}_1$ as
\begin{align*}
    &0\mapsto 00,\\
    &1\mapsto 01\ \mathrm{or}\ 10\ \mathrm{or}\ 11\ \mathrm{with\ the\ same\ probabilities}.
\end{align*}
We define the decoding rules in the same way for $\mu^{\pm}_2,\mu^{\pm}_3,\mu^{\pm}_4$.
In the magic state rounding algorithm for our space compression ratio preserving quantum relaxation, we choose one of the four measurement bases $\{\mu_i^{\pm}\}_{i\in[4]}$ and decode the encoded two bits according to the rule defined above.
Then, this rounding algorithm is equivalent to the depolarizing channel of $\lambda=\frac{7}{9}$.
\begin{lem}
    The magic state rounding algorithm for the space compression ratio preserving quantum relaxation described above is equivalent to applying $\Delta_{\frac{7}{9}}$ for all qubits where $\Delta_{\lambda}$ is a single qubit depolarizing channel defined by the equation
    \begin{equation*}
        \Delta_{\lambda}(\rho):=\lambda\cdot\frac{1}{2}I+(1-\lambda)\rho.
    \end{equation*}
\end{lem}
These equations mean that if we consider the expectation approximation ratio, the approximate solution's value obtained from our quantum relaxation can be written as
\begin{equation*}
    \mathrm{Tr}\left[\tilde{H}_{relax}\Delta_{\frac{7}{9}}^{\otimes n}(\tilde{\rho}_{relax})\right],
\end{equation*}
where $n$ is the number of the qubits involved.
We note that
\begin{equation}
    \begin{split}
        &\Delta_{\lambda}(P)=(1-\lambda)P\ (P\in\{X,Y,Z\}),\\
        &\Delta_{\lambda}(I)=I.
    \end{split}
    \label{eq:shrink}
\end{equation}

Our interest is the approximation ratio bound of our space compression ratio preserving quantum relaxation.
Unfortunately, we didn't obtain the constant expected approximation ratio for it.
Instead, we have the approximation ratio bound dependent on the ratio of the edges whose endpoints are associated with different qubit denoted by $\lambda\in[0,1]$ and the parameter $\epsilon\in\left[0,\frac{1}{2}\right]$ defined by the equation:
\begin{equation}
    OPT=\left(\frac{1}{2}+\epsilon\right)|E(G)|.
    \label{eq:gain}
\end{equation}
We note that $\epsilon$ is called the gain, and the problem to calculate the value $\epsilon$ is called MaxCutGain~\cite{CharikarWirth2004}.
Before proceeding to the proof of the approximation ratio, we prove the following lemma:
\begin{lem}
    Let $\rho$ be a $n$-qubit quantum state and let $P$ be a $k$-local Pauli operator $(0\leq k\leq n)$.
    Then, we have the bound:
    \begin{equation*}
        -1\leq\mathrm{Tr}[P\cdot\rho]\leq 1.
    \end{equation*}
    \label{lem:trace_range}
\end{lem}

\begin{thm}
    Let $\lambda\in[0,1]$ be the ratio of the edges whose endpoints are associated with different qubits.
    Let $\epsilon\in[0,\frac{1}{2}]$ be the gain defined in \Cref{eq:gain}.
    Consider an oracle $\tilde{\mathcal{O}}_{relax}$ which prepares the relaxed state for the space compression ratio preserving quantum relaxation using the encoding in \Cref{eq:tetra} satisfying the condition $\mathrm{Tr}[\tilde{H}_{relax}\tilde{\rho}_{relax}]\geq OPT$ where $OPT$ is the optimum value.
    Given access to $\tilde{\mathcal{O}}_{relax}$, the magic state rounding algorithm defined in this section solves the MaxCut problem with an expected approximation ratio
    \begin{equation*}
        \mathbb{E}[\gamma]\geq\max\left\{\frac{81-14\sqrt{3}+14\sqrt{3}\lambda+8\epsilon}{81+162\epsilon},\frac{27-14\lambda+12\epsilon}{27+54\epsilon}\right\}. 
    \end{equation*}
    \label{thm:tetra_bound}
\end{thm}
\begin{proof}
Let $n$ be the number of qubits involved in the algorithm.
We denote the magic rounding procedure by $\mathcal{M}$ as well as in \Cref{sec:qrao}.
By definition,
\begin{align*}
    \mathbb{E}[\gamma]&=\mathbb{E}\left[\frac{\mathrm{Tr}[\tilde{H}_{relax}\mathcal{M}^{\otimes n}(\tilde{\rho}_{relax})]}{OPT}\right]\\
    &=\frac{1}{OPT}\cdot\left(\frac{|E(G)|}{2}+\mathrm{Tr}\left[\left(\tilde{H}_{relax}-\frac{|E(G)|}{2}I^{\otimes n}\right)\cdot\Delta_{\frac{7}{9}}^{\otimes n}(\tilde{\rho}_{relax})\right]\right).
\end{align*}
The operator $\tilde{H}_{relax}-\frac{|E(G)|}{2}I^{\otimes n}$ can be divided into the sum of $1$-local operators and the sum of $2$-local operators by definition in \Cref{eq:tetra_hamiltonian} like the following:
\begin{equation*}
    \tilde{H}_{relax}-\frac{|E(G)|}{2}I^{\otimes n}=-\frac{3}{2}\sum_{\substack{(i,j)\in E(G):\\ \mathrm{Q\_idx(i)}\neq\mathrm{Q\_idx(j)}}}P_iP_j-\frac{\sqrt{3}}{2}\sum_{\substack{(i,j)\in E(G):\\ \mathrm{Q\_idx(i)}=\mathrm{Q\_idx(j)}=k}}Z_k.
\end{equation*}
From the assumption,
\begin{align*}
    &\mathrm{Tr}\left[\left(\tilde{H}_{relax}-\frac{|E(G)|}{2}I^{\otimes n}\right)\cdot\tilde{\rho}_{relax}\right]\\
    &=\mathrm{Tr}\left[\left(-\frac{3}{2}\sum_{\substack{(i,j)\in E(G):\\ \mathrm{Q\_idx(i)}\neq\mathrm{Q\_idx(j)}}}P_iP_j\right)\cdot\tilde{\rho}_{relax}\right]\\
    &\ \ \ \ +\mathrm{Tr}\left[\left(-\frac{\sqrt{3}}{2}\sum_{\substack{(i,j)\in E(G):\\ \mathrm{Q\_idx(i)}=\mathrm{Q\_idx(j)}=k}}Z_k\right)\cdot\tilde{\rho_{relax}}\right]\\
    &\geq OPT-\frac{|E(G)|}{2}.
\end{align*}
For simplicity, we define the variables $x$ and $y$ defined below:
\begin{align*}
    x&:=\frac{\mathrm{Tr}\left[\left(-\frac{3}{2}\sum_{\substack{(i,j)\in E(G):\\ \mathrm{Q\_idx(i)}\neq\mathrm{Q\_idx(j)}}}P_iP_j\right)\cdot\tilde{\rho}_{relax}\right]}{|E(G)|},\\
    y&:=\frac{\mathrm{Tr}\left[\left(-\frac{\sqrt{3}}{2}\sum_{\substack{(i,j)\in E(G):\\ \mathrm{Q\_idx(i)}=\mathrm{Q\_idx(j)}=k}}Z_k\right)\cdot\tilde{\rho_{relax}}\right]}{|E(G)|}.
\end{align*}
Then the assumption can be rewritten as
\begin{equation}
    x+y\geq\epsilon.
    \label{eq:assumption}
\end{equation}
The number of edges whose endpoints are associated with different (or the same) qubits is $|E(G)|\lambda$ (or $|E(G)|(1-\lambda)$). Combining this fact with Lemma 12, we have
\begin{align}
    &-\frac{3}{2}\lambda\leq x\leq\frac{3}{2}\lambda,\label{eq:xrange}\\
    &-\frac{\sqrt{3}}{2}(1-\lambda)\leq y\leq\frac{\sqrt{3}}{2}(1-\lambda).\label{eq:yrange}
\end{align}
From the self-adjointness of the single qubit depolarizing channel and the Lemma 11,
\begin{align*}
    \mathbb{E}[\gamma]&=\frac{1}{OPT}\cdot\left(\frac{|E(G)|}{2}+\mathrm{Tr}\left[\Delta_{\frac{7}{9}}^{\otimes n}\left(\tilde{H}_{relax}-\frac{|E(G)|}{2}I^{\otimes n}\right)\cdot\tilde{\rho}_{relax}\right]\right)\\
    &=\frac{|E(G)|}{OPT}\cdot\left(\frac{1}{2}+\frac{4}{81}x+\frac{2}{9}y\right)\\
    &=\frac{\frac{1}{2}+\frac{4}{81}(x+\frac{9}{2}y)}{\frac{1}{2}+\epsilon}.
\end{align*}
By minimizing $x+\frac{9}{2}y$ under \Cref{eq:xrange,eq:yrange,eq:assumption}, we have
\begin{equation*}
    x+\frac{9}{2}y\geq\max\left\{-\frac{7\sqrt{3}}{4}(1-\lambda)+\epsilon,-\frac{21}{4}\lambda+\frac{9}{2}\epsilon\right\}.
\end{equation*}
As a result, we obtain the expected approximation ratio bound:
\begin{equation*}
    \mathbb{E}[\gamma]\geq\max\left\{\frac{81-14\sqrt{3}+14\sqrt{3}\lambda+8\epsilon}{81+162\epsilon},\frac{27-14\lambda+12\epsilon}{27+54\epsilon}\right\}.
\end{equation*}
\end{proof}

Consider the condition of $\lambda$ and $\epsilon$ when our quantum relaxation has non-obvious approximation ratio $>\frac{1}{2}$.
\begin{align}
    &\mathbb{E}[\gamma]>\frac{1}{2}\\
    &\iff\left(\frac{81-14\sqrt{3}+14\sqrt{3}\lambda+8\epsilon}{81+162\epsilon}>\frac{1}{2}\right)\lor\left(\frac{27-14\lambda+12\epsilon}{27+54\epsilon}>\frac{1}{2}\right)\\
    &\iff\overline{\left(\frac{81-14\sqrt{3}+14\sqrt{3}\lambda+8\epsilon}{81+162\epsilon}\leq\frac{1}{2}\right)\land\left(\frac{27-14\lambda+12\epsilon}{27+54\epsilon}\leq\frac{1}{2}\right)}\\
    &\iff\begin{cases}0\leq\lambda\leq1&\mathrm{if}\ \epsilon<\frac{81-\sqrt{3}}{146+30\sqrt{3}}\approx0.4004\\
    0\leq\lambda<\frac{27}{28}-\frac{15}{14}\epsilon,-\frac{27\sqrt{3}}{28}+1+\frac{73\sqrt{3}}{42}\epsilon<\lambda\leq 1&\mathrm{if}\ 0.4004\approx\frac{81-\sqrt{3}}{146+30\sqrt{3}}\leq\epsilon\leq\frac{1}{2}\end{cases}\label{eq:lambda_range}
\end{align}
\Cref{fig:lambda_range} is the plot of the relation in \Cref{eq:lambda_range}.
The orange-colored part corresponds to the condition in \Cref{eq:lambda_range}.
The gray-colored part in the plot represents the condition of $\epsilon$ and $\lambda$ where our space compression ratio preserving quantum relaxation has no non-obvious approximation ratio bound.
If the graph instance has a relatively small MaxCut value (i.e. the gain $\epsilon<0.4004$), the space compression ratio preserving quantum relaxation has a non-trivial approximation ratio bound for arbitrary $\lambda$.
It means that we do not have to care about anything when assigning vertices to the qubits in the preprocessing.
\begin{figure}[tb]
    \begin{tabular}{cc}
        \begin{minipage}[t]{0.45\hsize}
            \centering
            \includegraphics[height=6cm]{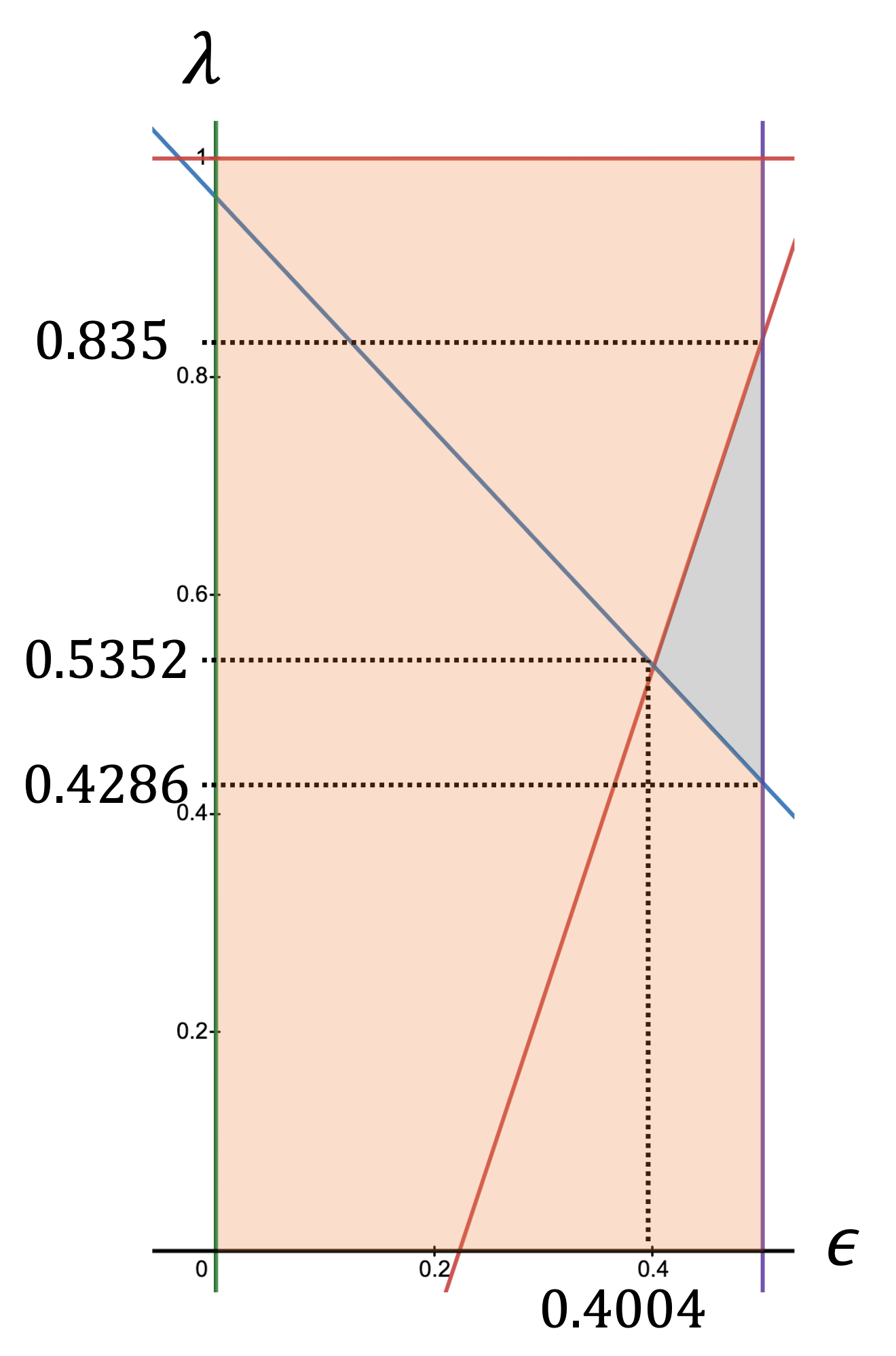}
            \subcaption{The plot of the range of $\lambda$ and $\epsilon$ with the non-obvious approximation ratio (orange-colored part). The horizontal axis corresponds to $\epsilon$ and the vertical axis corresponds to $\lambda$. The gray-colored part represents the condition that the space compression ratio preserving quantum relaxation has no meaningful approximation ratio bound.}
            \label{fig:lambda_range}
        \end{minipage} &
        \begin{minipage}[t]{0.45\hsize}
            \centering
            \includegraphics[height=6cm]{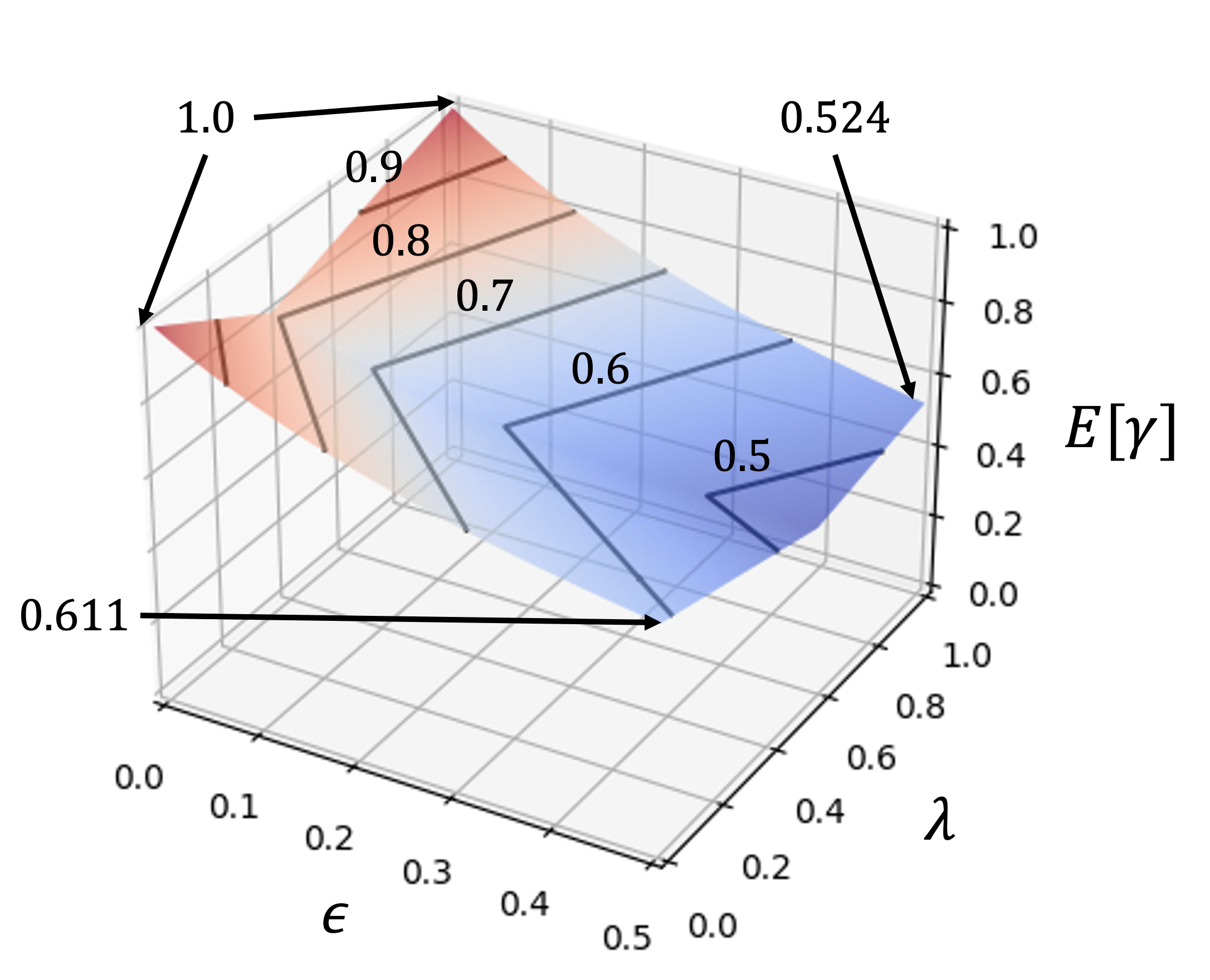}
            \subcaption{The 3D plot of the expected approximation ratio of the space compression preserving quantum relaxation for various $\epsilon$ and $\lambda$. Contours designate the values $\mathbb{E}[\gamma]=0.5, 0.6, 0.7, 0.8, 0.9$. The values of $\mathbb{E}[\gamma]$ for $\epsilon=0, 0.5$ and $\lambda=0, 1.0$ are also shown in the figure.}
            \label{fig:ratio_dist}
        \end{minipage}
    \end{tabular}
    \caption{The condition that the approximation ratio bound exceeds $\frac{1}{2}$ in the space compression ratio preserving quantum relaxation}
\end{figure}

By calculating the approximation ratio for each pair of $\epsilon$ and $\lambda$, we can analyze the performance of the algorithm more in detail.
\Cref{fig:ratio_dist} is the 3D plot of the approximation ratio bound in Theorem 12 for various $\epsilon$ and $\lambda$.
From the 3D plot, we have the following observations.
\begin{itemize}
    \item The approximation ratio becomes better as the bias between the number of edges encoded as the $1$-local Pauli operators and that of edges encoded as the $2$-local Pauli operators in the relaxed Hamiltonian $\tilde{H}_{relax}$.
    \item The approximation ratio bound is better for the instances whose gain $\epsilon$ is smaller.
\end{itemize}

\section{Conclusions}
\subsection{Summary of Results}
We theoretically extend the quantum relaxation in two ways.
Firstly, we extend the QRAO using the $(3,2)$-QRAC which encodes three classical bits into two qubits, i.e., achieving a 1.5 bit-to-qubit compression ratio. We proved the approximation ratio of the quantum relaxation using $(3,2)$-QRAC as $0.722$ which is better than the one for $(3,1)$- and $(2,1)$-QRAOs.
Secondly, we design a quantum relaxation whose bit-to-qubit compression ratio is always $2.0$ and independent of the density of the graph instances by using a novel encoding from two classical bits into a single qubit based on the formulation of $(3,1)$-QRAC.
We proved the approximation ratio of this space compression ratio preserving quantum relaxation by using the gain $\epsilon$ of the MaxCut and the ratio $\lambda$ of the edges whose endpoints are assigned to different qubits.
Though our approximation ratio bound $\max\left\{\frac{81-14\sqrt{3}+14\sqrt{3}\lambda+8\epsilon}{81+162\epsilon},\frac{27-14\lambda+12\epsilon}{27+54\epsilon}\right\}$ is not always larger than the obvious bound $\frac{1}{2}$, we analyze the condition that the bound becomes non-trivial $\left(\frac{1}{2}\right)$ and conclude that if the gain of the instance is not so large ($\epsilon<0.416$), our bound is larger than $\frac{1}{2}$.

\subsection{Future Directions}
\subsubsection{Approximation Ratio and Space Compression Ratio}
We consider the information-theoretic analysis of the trade-off between the approximation ratio and the space compression ratio of the quantum relaxation, which seems to contribute to revealing the theoretical limitation of the quantum-relaxation based approaches.
From the result of QRAO~\cite{fuller2021approximate} and our result of the quantum relaxation using $(3,2)$-QRAC, we conjectured the approximation ratio of the quantum-relaxation using a QRAC with the bit-to-qubit compression ratio $r$ as $\frac{1}{2}\left(1+r^{-2}\right)$:
\begin{conj}
    The expected approximation ratio of the quantum relaxation using $(m,n)$-QRAC for the MaxCut problem is conjectured to be
    \begin{equation*}
        \frac{1}{2}\left(1+\left(\frac{n}{m}\right)^2\right).
    \end{equation*}
    This approximation ratio bound assumes that the found relaxed state's energy exceeds the classical MaxCut value.
\end{conj}
Our space compression ratio preserving quantum relaxation is not included in the quantum relaxations mentioned in the above conjecture because it does not use the formulation of $(3,1)$-QRAC directly.
The difficulty of proving this conjecture lies in the point that the concrete formulations of QRACs for general $m$ and $n$ are not known.
The $(3,2)$-QRAC is obtained by numerical calculations, and it is hard to extend the rule of the construction of the QRAC to general $m$ and $n$.

\subsubsection{Hardness to Find Relaxed State Whose Energy Exceeds Classical MaxCut Value}
Also, considering the difficulty of the assumption of the proof of the approximation ratio of quantum relaxations, that is, the found relaxed state's energy exceeds the classical MaxCut value, is one of the future works.
Currently, searching for the relaxed state is executed by classical-quantum hybrid heuristics such as VQE~\cite{peruzzo2014variational}.
To analyze theoretically, we need to consider the approximability of the problem to find the maximum/ground eigenstate of the local Hamiltonian used in quantum relaxations.
We note again that to find the exact maximum/ground eigenstate is QMA-hard in general~\cite{kempe2006complexity}.
Also, because it is known that arbitrary quantum polynomial-time approximation algorithm cannot approximate the MaxCut problem with ratio $0.879+\epsilon$ for any $\epsilon>0$ under UGC and the assumption that polynomial hierarchy does not collapse~\cite{khot2002power,mossel2005noise,khot2007optimal,aaronson2010bqp}, at least at some space compression ratio $r$ satisfying $\frac{1}{2}\left(1+r^{-2}\right)\leq 0.879$ (see Conjecture 14), the assumption seems to become QMA-hard.
However not much is known about the limit of approximation for the local Hamiltonian problems even when we consider the classical approximation algorithms.

\subsubsection{Performance Analysis for MaxCutGain Problem}
As discussed in Appendix I\hspace{-1.2pt}I\hspace{-1.2pt}I of the original QRAO paper~\cite{fuller2021approximate}, quantum relaxation seems to perform well for the instances of the MaxCut problem with a small gain $\epsilon$.
Though the approximation ratio of the classical best approximation algorithm by Goemans and Williamson~\cite{goemans1995improved} is better than that of quantum relaxations, the GW algorithm's performance is bad for the instances with a small gain.
For instance, if $\epsilon=0.05$, GW outputs the cut of the size $OPT\times0.879=0.55\times|E(G)|\times0.879<0.5|E(G)|$.
The value $0.5|E(G)|$ is the trivial lower bound because randomly assigning binary variables generates the cut of the size $0.5|E(G)|$ in expectation.
On the other hand, quantum relaxation outputs the non-trivial cut of size
\begin{equation*}
    \left(0.5+\epsilon r^{-2}\right)|E(G)|>0.5|E(G)|
\end{equation*}
when $r>1$.
When we consider the MaxCutGain problem which evaluates the value of $\epsilon$ if the found relaxed state's energy exceeds the classical MaxCut value, quantum relaxation approximates it with the constant factor $r^{-2}$.
We note that there exists a classical approximation algorithm for the MaxCutGain problem which approximates the gain as $\Omega\left(\frac{\epsilon}{\log{\epsilon^{-1}}}\right)$~\cite{CharikarWirth2004} and this algorithm is known to be tight under UGC~\cite{v005a004}.
Then, some possibilities can be considered.
The first possibility is that finding a relaxed state whose energy exceeds the classical MaxCut value is also QMA-hard for arbitrary space compression ratio $r$.
The other possibility is that the analog of UGC does not hold in polynomial quantum time calculation.
There may be other possibilities as well.

\subsubsection{Other Candidates than QRACs}
Quantum random access codes are used in quantum relaxation to relax the classical problem and compress the classical bits.
However, we do not have to focus only on QRACs, and there are some other space compression approaches.
For example, symmetric informationally-complete positive operator-valued measurement (SIC-POVM)~\cite{renes2004symmetric} may be used to formulate the different types of quantum relaxations.
For $d$-dimensional Hilbert space, a POVM which consists of at least $d^2$ operators spanning the space of self-adjoint operators is called informationally complete POVM (IC-POVM).
We note that a mutually unbiased basis (MUB) is known to be IC-POVM~\cite{ivonovic1981geometrical,wootters1989optimal}.
Consider a set of $d^2$ projectors $\{\Pi_i\}_{i\in[d^2]}$ satisfying
\begin{equation*}
    \mathrm{Tr}[\Pi_i\Pi_j]=\frac{d\delta_{ij}+1}{d+1}
\end{equation*}
and set $F_i:=\frac{1}{d}\Pi_i$.
Then, $\{F_i\}_{i\in[d^2]}$ defines a minimal IC-POVM, and it is called SIC-POVM because of its symmetrical properties.
The encoding used in our space compression ratio preserving quantum relaxation $\{\tilde{\rho}_{x_1x_2}\}_{x_1,x_2\in\{0,1\}^2}$ forms the SIC-POVM $\{\frac{1}{4}\tilde{\rho}_{x_1x_2}\}_{x_1,x_2\in\{0,1\}^2}$.
For a general $n$-qubit system, there exists a SIC-POVM with $4^{n}$ operators.
Because of the use of POVMs, we have to take the overlaps of different operators into account, and it is known that a POVM can be converted into a projective measurement by introducing ancilla qubits and expanding the Hilbert space (Naimark's dilation)~\cite{naimark1943representation}.
However, the definition of the corresponding Hamiltonian is non-trivial even for the simplest SIC-POVM $\{\frac{1}{4}\tilde{\rho}_{x_1x_2}\}_{x_1,x_2\in\{0,1\}^2}$.

\section*{Acknowledgements}
R.R. would like to thank Sergey Bravyi and Pawel Wocjan of IBM T. J. Watson Research Center for the discussion on QRACs, as well as Yohichi Suzuki of Quantum Computing Center, Keio Univ. 
K.T. would like to thank Kaito Wada of Keio University for the discussion of future works about the use of SIC-POVM.
We would like to thank Ruho Kondoh for pointing out an error in the early version of this draft. 
\bibliographystyle{plain}
\bibliography{myref}

\begin{thebibliography}{10}

\bibitem{aaronson2010bqp}
Scott Aaronson.
\newblock Bqp and the polynomial hierarchy.
\newblock In {\em Proceedings of the forty-second ACM symposium on Theory of
  computing}, pages 141--150, 2010.

\bibitem{ambainis2002dense}
Andris Ambainis, Ashwin Nayak, Amnon Ta-Shma, and Umesh Vazirani.
\newblock Dense quantum coding and quantum finite automata.
\newblock {\em Journal of the ACM (JACM)}, 49(4):496--511, 2002.

\bibitem{CharikarWirth2004}
M.~Charikar and A.~Wirth.
\newblock Maximizing quadratic programs: extending {G}rothendieck's inequality.
\newblock In {\em 45th Annual IEEE Symposium on Foundations of Computer
  Science}, pages 54--60, 2004.

\bibitem{farhi2014quantum}
Edward Farhi, Jeffrey Goldstone, and Sam Gutmann.
\newblock A quantum approximate optimization algorithm.
\newblock {\em arXiv preprint arXiv:1411.4028}, 2014.

\bibitem{fuller2021approximate}
Bryce Fuller, Charles Hadfield, Jennifer~R Glick, Takashi Imamichi, Toshinari
  Itoko, Richard~J Thompson, Yang Jiao, Marna~M Kagele, Adriana~W
  Blom-Schieber, Rudy Raymond, et~al.
\newblock Approximate solutions of combinatorial problems via quantum
  relaxations.
\newblock {\em arXiv preprint arXiv:2111.03167}, 2021.

\bibitem{goemans1995improved}
Michel~X Goemans and David~P Williamson.
\newblock Improved approximation algorithms for maximum cut and satisfiability
  problems using semidefinite programming.
\newblock {\em Journal of the ACM (JACM)}, 42(6):1115--1145, 1995.

\bibitem{hayashi20064}
Masahito Hayashi, Kazuo Iwama, Harumichi Nishimura, Rudy Raymond, and Shigeru
  Yamashita.
\newblock (4, 1)-quantum random access coding does not exist—one qubit is not
  enough to recover one of four bits.
\newblock {\em New Journal of Physics}, 8(8):129, 2006.

\bibitem{holevo1979capacity}
Alexander~Semenovich Holevo.
\newblock On capacity of a quantum communications channel.
\newblock {\em Problemy Peredachi Informatsii}, 15(4):3--11, 1979.

\bibitem{imamichi2018constructions}
Takashi Imamichi and Rudy Raymond.
\newblock Constructions of quantum random access codes.
\newblock In {\em Asian Quantum Information Symposium (AQIS)}, volume~66, 2018.

\bibitem{ivonovic1981geometrical}
ID~Ivonovic.
\newblock Geometrical description of quantal state determination.
\newblock {\em Journal of Physics A: Mathematical and General}, 14(12):3241,
  1981.

\bibitem{iwama2007unbounded}
Kazuo Iwama, Harumichi Nishimura, Rudy Raymond, and Shigeru Yamashita.
\newblock Unbounded-error one-way classical and quantum communication
  complexity.
\newblock In {\em International Colloquium on Automata, Languages, and
  Programming}, pages 110--121. Springer, 2007.

\bibitem{kempe2006complexity}
Julia Kempe, Alexei Kitaev, and Oded Regev.
\newblock The complexity of the local {H}amiltonian problem.
\newblock {\em SIAM Journal on Computing}, 35(5):1070--1097, 2006.

\bibitem{khot2002power}
Subhash Khot.
\newblock On the power of unique 2-prover 1-round games.
\newblock In {\em Proceedings of the thiry-fourth annual ACM symposium on
  Theory of computing}, pages 767--775, 2002.

\bibitem{khot2007optimal}
Subhash Khot, Guy Kindler, Elchanan Mossel, and Ryan O’Donnell.
\newblock Optimal inapproximability results for {MAX-CUT} and other 2-variable
  {CSP}s?
\newblock {\em SIAM Journal on Computing}, 37(1):319--357, 2007.

\bibitem{v005a004}
Subhash Khot and Ryan O'Donnell.
\newblock Sdp gaps and ugc-hardness for max-cut-gain.
\newblock {\em Theory of Computing}, 5(4):83--117, 2009.

\bibitem{manvcinska2022geometry}
Laura Man{\v{c}}inska and Sigurd~AL Storgaard.
\newblock The geometry of {B}loch space in the context of quantum random access
  codes.
\newblock {\em Quantum Information Processing}, 21(4):1--16, 2022.

\bibitem{mossel2005noise}
Elchanan Mossel, Ryan O'Donnell, and Krzysztof Oleszkiewicz.
\newblock Noise stability of functions with low influences: invariance and
  optimality.
\newblock In {\em 46th Annual IEEE Symposium on Foundations of Computer Science
  (FOCS'05)}, pages 21--30. IEEE, 2005.

\bibitem{naimark1943representation}
Mark~A Naimark.
\newblock On a representation of additive operator set functions.
\newblock In {\em Dokl. Akad. Nauk SSSR}, volume~41, pages 373--375, 1943.

\bibitem{nannicini2019performance}
Giacomo Nannicini.
\newblock Performance of hybrid quantum-classical variational heuristics for
  combinatorial optimization.
\newblock {\em Physical Review E}, 99(1):013304, 2019.

\bibitem{nayak1999optimal}
Ashwin Nayak.
\newblock Optimal lower bounds for quantum automata and random access codes.
\newblock In {\em 40th Annual Symposium on Foundations of Computer Science
  (Cat. No. 99CB37039)}, pages 369--376. IEEE, 1999.

\bibitem{peruzzo2014variational}
Alberto Peruzzo, Jarrod McClean, Peter Shadbolt, Man-Hong Yung, Xiao-Qi Zhou,
  Peter~J Love, Al{\'a}n Aspuru-Guzik, and Jeremy~L O’brien.
\newblock A variational eigenvalue solver on a photonic quantum processor.
\newblock {\em Nature communications}, 5(1):1--7, 2014.

\bibitem{renes2004symmetric}
Joseph~M Renes, Robin Blume-Kohout, Andrew~J Scott, and Carlton~M Caves.
\newblock Symmetric informationally complete quantum measurements.
\newblock {\em Journal of Mathematical Physics}, 45(6):2171--2180, 2004.

\bibitem{teramoto2023role}
Kosei Teramoto, Rudy Raymond, and Hiroshi Imai.
\newblock The role of entanglement in quantum-relaxation based optimization
  algorithms.
\newblock {\em arXiv preprint arXiv:2302.00429}, 2023.

\bibitem{welsh1967upper}
Dominic~JA Welsh and Martin~B Powell.
\newblock An upper bound for the chromatic number of a graph and its
  application to timetabling problems.
\newblock {\em The Computer Journal}, 10(1):85--86, 1967.

\bibitem{wootters1989optimal}
William~K Wootters and Brian~D Fields.
\newblock Optimal state-determination by mutually unbiased measurements.
\newblock {\em Annals of Physics}, 191(2):363--381, 1989.

\bibitem{https://doi.org/10.48550/arxiv.2301.01778}
Andrew Zhao and Nicholas~C. Rubin.
\newblock Quantum relaxation for quadratic programs over orthogonal matrices,
  2023.

\end{thebibliography}

\end{document}